\documentclass[acmsmall,screen,nonacm]{acmart}

\usepackage{graphicx}
\usepackage{physics}
\usepackage{comment}
\usepackage{xparse}
\usepackage{xfrac}
\usepackage{galois}
\usepackage{subcaption}
\usepackage{mathbbol}
\usepackage{bbold}
\usepackage{stmaryrd}
\usepackage{algorithm}
\usepackage{algpseudocode}

\usepackage{scalerel}
\usepackage{adjustbox}
\usepackage{thm-restate}
\usepackage{tikz}
\usetikzlibrary{quantikz2}

\AtBeginDocument{%
  }

\setcopyright{rightsretained}
\acmDOI{}
\acmYear{2025}

\author{Nicola Assolini}
\authornote{These authors contributed equally to this work.}
\email{nicola.assolini@univr.it}
\author{Luca Marzari}
\authornotemark[1]
\email{luca.marzari@univr.it}
\author{Isabella Mastroeni}
\email{isabella.mastroeni@univr.it}
\author{Alessandra Di Pierro}
\email{alessandra.dipierro@univr.it}
\affiliation{%
  \institution{Department of Computer Science, University of Verona}
  \country{Italy}
}

\title{Formal Verification of Variational Quantum Circuits}
\date{March 2025}

\begin{document}
\newcommand{\TODO}[1][]{{\color{red}[TODO:#1]}}
\newcommand{\ISA}[1][]{{\color{blue}[ISA:#1]}}

\newcommand{\ii}{\iota}
\newcommand*{\arrw}[1]{\mathbf{#1}}
\newcommand{\sset}[2]{\left\{~#1  \left |
                    \begin{array}{l}#2\end{array} \right.\!\right\}}

\newcommand{\tuple}[1]{\langle#1\rangle}
\newcommand{\rtuple}[1]{<\!#1\!>} 
\newcommand{\btuple}[1]{\{#1\}}
\newcommand{\interv}[2]{[#1,#2]}
\newcommand{\opints}[1]{{#1}^{\mbox{\tiny $\sharp$}}}
\newcommand{\opint}[1]{\overline{#1}}
\newcommand{\tinyd}[1]{\mbox{\tiny $#1$}}

\def\bq{{\bf q}}

\newcommand{\bA}{\mathbb{A}}
\newcommand{\bB}{\mathbb{B}}
\newcommand{\bC}{\mathbb{C}}
\newcommand{\bD}{\mathbb{D}}
\newcommand{\bE}{\mathbb{E}}
\newcommand{\bF}{\mathbb{F}}
\newcommand{\bG}{\mathbb{G}}
\newcommand{\bH}{\mathbb{H}}
\newcommand{\bI}{\mathbb{I}}
\newcommand{\bJ}{\mathbb{J}}
\newcommand{\bK}{\mathbb{K}}
\newcommand{\bL}{\mathbb{L}}
\newcommand{\bM}{\mathbb{M}}
\newcommand{\bN}{\mathbb{N}}
\newcommand{\bO}{\mathbb{O}}
\newcommand{\bP}{\mathbb{P}}
\newcommand{\bQ}{\mathbb{Q}}
\newcommand{\bR}{\mathbb{R}}
\newcommand{\bS}{\mathbb{S}}
\newcommand{\bT}{\mathbb{T}}
\newcommand{\bU}{\mathbb{U}}
\newcommand{\bV}{\mathbb{V}}
\newcommand{\bW}{\mathbb{W}}
\newcommand{\bX}{\mathbb{X}}
\newcommand{\bY}{\mathbb{Y}}
\newcommand{\bZ}{\mathbb{Z}}
\newcommand{\bRI}{\mathbb{RI}}
\newcommand{\bCI}{\mathbb{CI}}

\def\defi{\stackrel{\mbox{\tiny def}}{=}}
\def\ok#1{\mbox{\raisebox{0ex}[1ex][1ex]{$#1$}}}

\def\cA{\mathcal{A}}
\def\cB{\mathcal{B}}
\def\cC{\mathcal{C}}
\def\cD{\mathcal{D}}
\def\cE{\mathcal{E}}
\def\cF{\mathcal{F}}
\def\cG{\mathcal{G}}
\def\cH{\mathcal{H}}
\def\cI{\mathcal{I}}
\def\cJ{\mathcal{J}}
\def\cK{\mathcal{K}}
\def\cL{\mathcal{L}}
\def\cM{\mathcal{M}}
\def\cN{\mathcal{N}}
\def\cO{\mathcal{O}}
\def\cP{\mathcal{P}}
\def\cQ{\mathcal{Q}}
\def\cR{\mathcal{R}}
\def\cS{\mathcal{S}}
\def\cT{\mathcal{T}}
\def\cU{\mathcal{U}}
\def\cV{\mathcal{V}}
\def\cW{\mathcal{W}}
\def\cX{\mathcal{X}}
\def\cY{\mathcal{Y}}
\def\cZ{\mathcal{Z}}

\def\tA{\mathtt{A}}
\def\tL{\mathtt{L}}
\def\tD{\mathtt{D}}
\def\lra{\to}
\def\Lra{\Leftrightarrow}
\newcommand{\fixpointsOf}[1]{\mathsf{fp}(#1)}
\def\func{\mathsf{f}}
\def\gunc{\mathsf{g}}

\newcommand*{\Hilb}{\mathcal{H}}
\newcommand*{\code}[1]{\mathtt{#1}}
\newcommand*{\gaten}{\mathtt{U}}
\newcommand*{\gateg}[1]{\mathtt{#1}}
\newcommand*{\nonter}[1]{\textsf{#1}}
\newcommand*{\CN}{\code{CX}}
\newcommand*{\qnn}{\nonter{c}}
\def\uco{\mathsf{uco}}
\def\prop{\Pi}

\newcommand*{\gateset}{\cU}
\newcommand*{\rotset}{\cR}
\newcommand*{\varset}[1][\qnn]{\code{QV}_{#1}}
\newcommand*{\parset}[1][\qnn]{\code{CV}_{#1}}
\newcommand*{\bases}[1][\qnn]{\cB_{#1}}
\newcommand*{\fun}[1][\ensuremath{x}]{\code{Fun}[#1]}

\newcommand*{\SpaceP}[1][\varset]{\Hilb_{#1}}
\newcommand*{\DomP}[1][\qnn]{\mathbb{\Psi}_{#1}} 
\newcommand*{\DomR}[1][\qnn]{\mathbb{\Delta}_{#1}} 
\newcommand*{\aDomP}[1][\qnn]{\opint{\bV}_{\!#1}} 
\newcommand*{\aDomR}[1][\qnn]{\opint{\bD}_{#1}} 
\newcommand*{\DomSenv}{\mathbb{\Sigma}_{\qnn}}
\newcommand*{\aDomSenv}{\opint{\bS}_{\qnn}}

\newcommand{\sSem}[2][\Stenv]{\llbracket #2\rrbracket_{#1}}
\newcommand{\csSem}[2][\Sigma]{\llbracket #2\rrbracket_{#1}}
\newcommand{\abSem}[2]{\opints{\llbracket #1\rrbracket}_{\mbox{\tiny $#2$}}}

\newcommand{\lang}{\cL_{\mbox{\tiny $\code{VQC}$}}}
\newcommand{\clipC}{\text{clip}_{\mbox{\tiny $\code{CI}$}}}
\newcommand{\clipR}{\text{clip}_{\mbox{\tiny $\code{RI}$}}}
\newcommand{\winner}{\xi}
\newcommand{\winnerc}{\xi^{\mbox{\tiny $\uparrow$}}}
\newcommand{\awinner}{\opints{\xi}}
\newcommand{\nat}{\mathbb{N}}
\newcommand*{\U}{\code{U}}
\newcommand{\M}{\code{M}}
\newcommand{\abM}{\opints{\Mf}}
\newcommand{\Up}[1][x]{\code{E}}
\newcommand{\abste}[1][V]{\opint{\code{#1}}}
\newcommand*{\Stenv}{\sigma}
\newcommand*{\Stenvc}{\code{S}}
\newcommand{\abenv}{\opint{\code{S}}}
\newcommand{\absdist}{\opints{\varrho}}
\newcommand{\dist}{\rho}
\newcommand{\distc}{\varrho}
\newcommand{\init}[1]{{{#1}_{\tinyd{0}}}}

\def\grass#1{\{\!|#1|\!\}}
\def\grassto#1{(\!|#1|\!)}
\def\agrassto#1{\opints{(\!|#1|\!)}}

\newcommand{\alphaR}{\alpha_{\mbox{\tiny $\code{RI}$}}}
\newcommand{\gammaR}{\gamma_{\mbox{\tiny $\code{RI}$}}}
\newcommand{\alphaC}{\alpha_{\mbox{\tiny $\code{CI}$}}}
\newcommand{\gammaC}{\gamma_{\mbox{\tiny $\code{CI}$}}}
\newcommand{\alphaP}{\alpha_{\mbox{\tiny $\code{st}$}}}
\newcommand{\gammaP}{\gamma_{\mbox{\tiny $\code{st}$}}}
\newcommand{\alphaE}{\alpha_{\mbox{\tiny $\code{env}$}}}
\newcommand{\gammaE}{\gamma_{\mbox{\tiny $\code{env}$}}}
\newcommand{\alphaD}{\alpha_{\mbox{\tiny $\code{dist}$}}}
\newcommand{\gammaD}{\gamma_{\mbox{\tiny $\code{dist}$}}}
\newcommand{\alphaRn}{\alpha_{\mbox{\tiny $\code{RI}^n$}}}

\def\comp{\mathrel{\hbox{\footnotesize${}\!{\circ}\!{}$\normalsize}}}

\newcommand{\bigboxplus}{\mathop{\scalerel*{\boxplus}{\bigoplus}}}

\begin{abstract}
    Variational quantum circuits (VQCs) are a central component of many quantum machine learning algorithms, offering a hybrid quantum-classical framework that, under certain aspects, can be considered similar to classical deep neural networks. A shared aspect is, for instance, their vulnerability to adversarial inputs—small perturbations that can lead to incorrect predictions. While formal verification techniques have been extensively developed for classical models, no comparable framework exists for certifying the robustness of VQCs. Here, we present the first in-depth theoretical and practical study of the formal verification problem for VQCs. Inspired by abstract interpretation methods used in deep learning, we analyze the applicability and limitations of interval-based reachability techniques in the quantum setting. 
    We show that quantum-specific aspects, such as state normalization, introduce inter-variable dependencies that challenge existing approaches. We investigate these issues by introducing a novel semantic framework based on abstract interpretation, where the verification problem for VQCs can be formally defined, and its complexity analyzed. Finally, we demonstrate our approach on standard verification benchmarks.
\end{abstract}

\maketitle

\section{Introduction}

The computational paradigm of Variational Quantum Circuits (VQCs) provides a general framework for solving a wide range of problems by employing the power of quantum computers in synergy with the most advanced classical optimisation tools~\cite{havlicek_supervised_2019,cerezo2021variational,endo2021hybrid}.
VQCs are algorithms inspired by one of the most fundamental principles in Quantum Mechanics, namely the Variational Principle, at the base of any feasible tool for approximating the energy of the ground state of a system. 
These algorithms are characterized by a modular design, with components that may vary in structure and complexity depending on the specific problem. For their versatility, they have been used to address problems in quantum chemistry, such as calculating molecular energies, in mathematical applications, such as solving systems of equations, in combinatorial optimization, error correction, circuit compilation, etc. VQCs also play a fundamental role in several techniques developed in Quantum Machine Learning~\cite{schuld2015introduction,biamonte2017quantum}, for classification tasks and, in general, for learning complex patterns from data. Examples of VQC-based approaches to machine learning include Variational Quantum Support Vector Machines~\cite{havlicek_supervised_2019} and Quantum Generative Adversarial Networks~\cite{qgan}.

Given the widespread use of VQCs in so many application fields, the problem of verifying their properties, such as robustness (also referred to as stability) and accuracy, is crucial. This problem arises from different aspects of the algorithms and their implementations. 
One aspect is related to the nature of the quantum device used for computation, allowing factors like decoherence and gate errors to introduce errors in the computation and measurement process. These errors propagate throughout the optimization process, impacting the accuracy of the results. Another aspect, which is closer to the specific application employing a VQC, is related to the heuristic nature of these algorithms and the need for an empirical evaluation to assess their properties. 

In this paper, we will address the problem of the verification of VQCs relatively to the second aspect mentioned above, with machine learning as a reference application and, in particular, the VQC-based classification method for quantum supervised learning.

To this purpose, a useful source of inspiration is the analogy with a popular method in classical machine learning, the Neural Network model (NN), and in particular with its training process, which is similar to the variational behaviour of a VQC. In fact, at the core of a VQC lies a parameterized quantum circuit with a specific structure, known as \emph{ansatz} or \emph{variational form}, which is formed by a sequence of parameterized gates and provides a trial (starting) point. This is followed by a training procedure that iteratively updates the parameters based on the classical optimization of a problem-specific cost function.
This model clearly resembles that of a neural network architecture, in the sense that artificial neural networks, too, are parameterized mathematical models, where parameters are updated during training by means of an optimization procedure.
Moreover, VQCs are, like NNs, susceptible to adversarial inputs, i.e., small perturbations to the input state leading to unexpected predictions~\cite{adversarial,wendlinger2024comparative}. 
Crucially, due to the heuristic and approximate nature of (quantum) training algorithms, the common approach for assessing their robustness is an empirical evaluation, which is obviously insufficient to provide provable guarantees. 

Based on this analogy between VQC-based supervised learning and classical neural networks, and inspired by some methods that have been developed for the formal verification of NNs~\cite{liu2021algorithms}, we present a framework for the theoretical and practical analysis of the problem of formally verifying variational quantum circuits.

We focus, in particular, on the investigation of one of the main challenges to the realization of formal verification methods, which, contrary to empirical evaluation, are able to provide provable guarantees, namely the complexity of the verification problem. As for classical NN, we will show in the paper that this problem can be, in general, computationally infeasible when the input space is very large, even if it is discrete and finite, as is the case in most real-world scenarios.

To address this problem, we will adopt the abstract interpretation framework and present a technique for conservatively approximating the behavior of a quantum circuit for the purpose of verifying the accuracy of a VQC-based classifier via the analysis of the approximation error (precision of the abstraction).

At first glance, an abstract approach may appear to be a simpler problem than in the classical case, where the verification of a deep NN has to deal with non-linear activation functions, which typically transform abstract domains (such as polytopes or hyperrectangles) into non-convex shapes that are difficult to analyze. 
In contrast, the training of a VQC is based on a purely linear learning process based on unitary (and thus convex) operators, represented by unitary matrices, possibly decomposed in a sequence of unitary operators, giving the circuit a layered structure.
However, this unitary process can only be applied after transforming the classical input data into quantum states by means of an appropriate encoding phase~\cite{rath_quantum_2023}.
If non-linear training does not represent a problem in our setting, we have to deal, nevertheless, with another component of a VQC, which is responsible for transforming classical input data into a quantum state. 
This transformation, called quantum embedding or quantum feature map, is typically achieved 
via a (highly) non-linear encoding on the set of inputs~\cite{schuld2021effect,havlicek_supervised_2019}, and changes the computational arena from the classical dataset to a quantum state space, i.e. a space of complex vectors of norm $1$.
Normalization introduces strong dependencies among data, which are particularly problematic for the definition of reasonable abstractions. 

While in classical deep NNs such dependencies are only a potential issue, in quantum systems this becomes particularly problematic.
On top of this issue, another challenge in VQCs is the \textit{measurement} operation performed at the end of the circuit to read the results of the computation. This operation produces a probability distribution over all possible outcomes and is obviously non-invertible: given a probability distribution, there are infinitely many quantum states that could produce it upon measurement.
As for non-linear neural networks, this aspect does not allow to perform a backward analysis and determine the sets of inputs that generated the desired output.

All these issues motivate the need for a novel formulation and analysis tailored specifically to VQCs. 
%

\paragraph{Contributions and Structures}
After recalling some basic notions of quantum computing, abstract interpretation, and VQC (\autoref{sec: background} and \autoref{sec: vqc}), the paper's contributions are as follows.
We formally define a {\bf language modeling VQCs} (syntax and semantics), based on a representation of quantum states as functions (\autoref{sec: vqc lan}).
We introduce an {\bf abstract semantics for the VQC language on intervals} of inputs (real values) and of amplitudes (complex values) in \autoref{sec: abstract}. Such a semantics provides a formal framework for reasoning on the sources of imprecision with the aim of reducing the uncertainty due to the approximation.
We analyze how our semantics, built on a non-relational domain, loses precision in a relational setting like normalized quantum states, highlighting its incompleteness and {\bf identifying sufficient conditions for completeness} (\autoref{sec: precision}).
In \autoref{sec: vqc classifier}, we define a {\bf VQC-based classifier} in both the concrete and abstract settings, formally characterizing the function that performs classification based on the quantum computation results. 
Building on this definition, we define the verification problem for VQC-based classifiers, proving its {\bf hardness}.
In \autoref{sec: recovery} we propose some techniques to {\bf mitigate the loss of precision} during verification.
In \autoref{sec: evaluation}, we {\bf demonstrate the impact of our new
verification framework} by applying it to verifying the robustness of several VQC-based classifiers on standard benchmarks also used in~\cite{lu2020quantum, wendlinger2024comparative}, namely \textit{Iris}~\cite{iris} and \textit{MNIST}~\cite{MNIST} datasets. 
To the best of our knowledge, our verification pipeline is the first tool able to provably certify the maximum $\epsilon$ input perturbation that is tolerable by a given VQC.

\section{Background}\label{sec: background}

\subsection{Quantum Computation}

In this section, we briefly recall the main aspects of quantum computation that make this computational model different from the classical one.
For more details, we refer to~\cite{MichealANielsen}.
In the circuit model of computation, the first difference between the two computational paradigms is that in a quantum circuit, wires represent quantum bits, or qubits, rather than bits.

According to the postulates of Quantum Mechanics, the state of one qubit is any vector of norm $1$ in a two-dimensional complex Hilbert space, and the state of a register of $n$ qubits is a vector of norm $1$ in the  $2^n$-dimensional complex Hilbert space corresponding to the tensor product of the 2-dimensional complex Hilbert space of each qubit. The complex coefficients are called probability amplitudes for their role in the measurement process described later.

We will use the standard notation in quantum physics and indicate a vector $(\alpha_1,\ldots,\alpha_m)$ in a m-dimensional Hilbert space with respect to a basis $\cB$  by the Dirac ket $\ket{\psi}=\sum_{k=1}^m \alpha_k \ket{e_k}$, where $\ket{e_k}\in\cB$. Thus, for example, the standard\footnote{Every set of $m$ orthonormal vectors is a basis for an m-dimensional Hilbert space. The standard basis is used as a default basis when no other definition is given.} basis vectors of the state space of one qubit will be denoted by $\ket{0}$ and $\ket{1}$ corresponding to $[1,0]^T$ and $[0,1]^T$, respectively, and a generic vector, $[\alpha,\beta]^T$, in this space by $\alpha \ket{0} + \beta \ket{1}$.
The postulate about the composition of quantum systems by tensor product allows us to construct the standard basis for the space of two qubits by 
$\{\ket{0},\ket{1}\}\otimes \{\ket{0},\ket{1}\} = \{\ket{0}\otimes\ket{0},\ket{0}\otimes\ket{1},\ket{1}\otimes\ket{0},\ket{1}\otimes\ket{1}\}$, also denoted by 
$\{\ket{00},\ket{01},\ket{10},\ket{11}$, corresponding to (in linear algebraic notation) $\{[1,0,0,0]^T,[0,1,0,0]^T,[0,0,1,0]^T,[0,0,0,1]^T\}$. By using the same procedure, we can construct the standard basis for an m-dimensional Hilbert space.

An example of a quantum circuit is shown in \autoref{fig: circuit_ex}, where the two initial gates represent two particular unitary operators on one qubit, namely $\gateg{H}$, which transforms a basis vector into a superposition, and $\gateg{X}$, which is the unitary implementation of the classical NOT gate. 
All operations on $1$ qubit correspond to rotations of the point in the unitary three-dimensional sphere (the so-called Bloch sphere) with polar coordinates corresponding to the qubit's amplitudes. The rotations, $\gateg{Rx}$, $\gateg{Ry}$, $\gateg{Rz}$, around the three axes $x,y,z$ are universal for all operations on one qubit, in this case $\gateg{H}$ is a rotation of a $\sfrac{\pi}{2}$ angle around the $y$ axes while the NOT is a rotation of $\pi$ around $y$ or $x$. 
The next gate in the circuit is the controlled-NOT ($\CN$) operation, which applies $\gateg{X}$ to the second qubit only if the first qubit is $\ket{1}$.
This circuit implements the unitary operator $\CN\cdot (\gateg{X}\otimes \gateg{H})$.
Here, the tensor product $(\gateg{X}\otimes \gateg{H})$ represents the operation $\gateg{H}$ applied to the less significant qubit ($q_0$) and $\gateg{X}$ to the most significant one, $q_1$. 
If the initial state is $\ket{00}$, where bot $q_0$ and $q_1$ are $\ket{0}$, $\ket{\psi_1} = (\gateg{X}\otimes \gateg{H})\cdot\ket{00} = \ket{1}\sfrac{1}{\sqrt{2}}(\ket{0}+\ket{1})$, i.e. we have negated $q_1$ and applied $\gateg{H}$ to $q_0$. 
The final state is $\ok{\CN\ket{\psi_1} = \frac{1}{\sqrt{2}}(\ket{01}+\ket{10})}$.
All quantum circuits can be constructed by using only gates in $\{\gateg{Rx},\gateg{Rx},\gateg{Rx}, \CN\}$, i.e., this set is universal for quantum computation\cite{barenco95elementary}.
%
The final gate 
is the measurement operation, which is necessary to extract the final (classical) result from the quantum state $\ket{\psi_2}$ obtained from the quantum evolution of the input $\ket{00}$.
Formally, quantum measurement corresponds to the application of projection operators\footnote{We will consider here only a special type of quantum measurement called `projective measurement'.} 
$\{M_e\}_{\{e \in O\}}$, where $O$ is the set of all labels (possible outcomes) associated with the basis 
$B=\{\ket{e}\}_{e\in O}$ of the state space.
Thus, if the system is in the quantum state $\ket{\psi}$ before the measurement, the probability of obtaining the outcome $e$ is given by $p(e) = \|M_e\ket{\psi}\|^2$, i.e. the squared modulus of the probability amplitude (coefficient) of $\ket{e}$ in $\psi=\sum_{e\in O}\alpha_e\ket{e}$.
As an example, to measure a qubit in the standard basis 
$\{\ket{0},\ket{1}\}$ means to apply the projector operators $M_0$ on the axis $\ket{0}$ and $M_1$ on the axis $\ket{1}$. 
If the qubit state before the measurement is $\ket{\psi} = \alpha \ket{0} + \beta \ket{1}$, the probability of measuring $0$ is $p(0) = \|M_0 \ket{\psi}\|^2 = |\alpha|^2$, and the probability of measuring $1$ is $p(1) = \|M_1 \ket{\psi}\|^2 = |\beta|^2$. 

\begin{figure}
    \centering
    \resizebox{.3\linewidth}{!}{
    \begin{quantikz}[row sep=2pt]
        \lstick{$q_0: \ket{0}$} & \gate{\gateg{H}} & \ctrl{1} & \meter{}\\
        \lstick{$q_1: \ket{0}$} & \gate{\gateg{X}} & \targ{} & \meter{}
    \end{quantikz}
    }
    \caption{A simple quantum circuit implementing the unitary operator $\CN\cdot (\gateg{X}\otimes \gateg{H})$ applied to $\ket{00}$.} \label{fig: circuit_ex}
\end{figure}
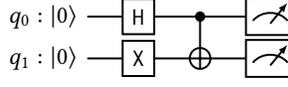

\subsection{Abstract Interpretation in a Nutshell}\label{sec: abint}
In the standard framework of \emph{abstract interpretation}~\cite{cousot1977abstract}, abstract domains can be equivalently formulated either in terms of Galois connections or closure operators~\cite{cousot1979systematic}. A pair of functions $\alpha : \tL \lra \tA$ and $\gamma : \tA \lra \tL$ on posets forms an \emph{adjunction} or a \emph{Galois insertion} (GI, denoted $\ok{\tL\galoiS{\alpha}{\gamma}\tA}$) if for any $x \in \tL$ and $y \in \tA$ we have $\alpha(x) \leq_{\tA} y \Lra x \leq_\tL \gamma(y)$. In this case, $\alpha$ (resp. $\gamma$) is the \emph{left-} (resp. \emph{right-}) \emph{adjoint} of $\gamma$ (resp. $\alpha$), and it is additive (resp. co-additive). Given a co-additive\footnote{A function is co-additive if it commutes with the greatest lower bound, dually an additive function commutes with the least upper bound.} function $\gamma$ we can always build its adjoint as $\alpha(x)=\bigwedge_{\tA}\sset{y\in\tA}{x\leq_{\tL}\gamma(y)}$, and vice-versa, given any additive function $\alpha$ it is always possible to build its adjoint.
%
Abstract interpretation is generally used for abstracting {\em functions}. 
Given a function $\func: \tL_i \lra \tL_o$, an input abstraction $\ok{\tL_i\galoiS{\alpha_i}{\gamma_i}\tA_i}$ and an output abstraction $\ok{\tL_o\galoiS{\alpha_o}{\gamma_o}\tA_o}$, we say that $\ok{\func^\sharp}:\tA_i\to\tA_o$ is a sound abstraction of $\func$ if $\forall x\in\tL_i.\:\alpha_o \comp \func(x)\leq_{\tL_o} \ok{\func^\sharp}\comp\alpha_i(x)$, meaning that we may lose information when computing abstract inputs. The abstract function is complete if 
$\forall x\in\tL_i.\:\alpha_o \comp \func(x)= \ok{\func^\sharp}\comp\alpha_i(x)$, i.e., when there is no loss of precision by computing with objects in the abstract domain~\cite{cousot1977abstract, cousot1979systematic}. If completeness holds for a fixed sets of inputs $X$ we say that $\func^\sharp$ is local complete on $X$~\cite{Bruni2021AInterpretations,imvmcai25} The best of such approximation is called best correct approximation (bca for short), and it is defined as $\alpha_o \comp \func \comp \gamma_i$~\cite{cousot1979systematic}.
Note that if there exists a $\func^\sharp$ that is complete, then also the bca is complete~\cite{cousot1977abstract, cousot1979systematic, GRS00}.

\section{Variational Quantum Circuits}\label{sec: vqc}
A Variational Quantum Circuit (VQC) 
consists of three main components: an input encoding (feature map or quantum embedding), a parametric circuit, and a final measurement followed by a classical optimization of a cost function for the parameters' update.
We describe these components via the small circuit on two qubits depicted in \autoref{fig: VQC_example}. 
In this example, the encoding consists of two gates $\gateg{R_x}[x_i]$, $i\in\{0,1\}$, which are applied to the qubits $q_0$ and $q_1$, both in the blank state $\ket{0}$.
The initial state is $\ket{00} = \ket{0}_{q_1} \otimes \ket{0}_{q_0}$.
These gates are rotations around the $x$-axis by an angle that is proportional to the input features $x_i$ of the data and, for a generic $x$, are therefore defined by 
\begin{gather}\label{eq: Rx}
    \gateg{Rx}[x] =
    \begin{pmatrix} \cos(\sfrac{x}{2}) & -\ii \sin(\sfrac{x}{2}) \\
    -\ii \sin(\sfrac{x}{2}) & \cos(\sfrac{x}{2})
\end{pmatrix},
\end{gather}
We use the notation $\gateg{Rx}^{q_0}[x_0]$ for the operator $\gateg{I}\otimes\gateg{Rx}[x_0]$, where $\gateg{I}$ is the identity operator, acting as a rotation on $q_0$ and as the identity on $q_1$.
Thus, $\gateg{Rx}^{q_0}[x_0]$ is a $4\times 4$ matrix that acts on the whole vector space of the circuit.
In our example, the encoding can be obtained by computing:
\begin{equation*}
    \begin{aligned}
        \gateg{Rx}^{q_1}[x_1] \cdot \gateg{Rx}^{\tinyd{q_0}}[x_0] \cdot \ket{00} =& \cos(\sfrac{x_1}{2}) \cos(\sfrac{x_0}{2})\ket{00} + \ii \cos(\sfrac{x_1}{2}) \sin(\sfrac{x_0}{2})\ket{01} + \\
        &+ \ii \sin(\sfrac{x_1}{2}) \cos(\sfrac{x_0}{2})\ket{10} - \sin(\sfrac{x_1}{2}) \sin(\sfrac{x_0}{2})\ket{11}
    \end{aligned}
\end{equation*}
%
This shows how the encoding gates map the values of $x_i$  in the amplitude of a quantum state, on which we can now operate with the parametric quantum circuit.
\begin{figure}[t]
    \centering
    \resizebox{.5\linewidth}{!}{\begin{quantikz}[row sep=2pt, slice all, slice label style={inner sep=1pt,anchor=south west,rotate=40}, slice
titles=$\ \ket{\psi_{\col}}$]
        \lstick{$q_0: \ket{0}$} & \gate{\gateg{Rx}[x_0]}\gategroup[2,style={dashed,rounded
            corners,fill=blue!20, inner
            sep=0.5pt},background,label style={label position=below, anchor=north,yshift=-0.2cm}]{Encoding} &  \gate{\gateg{Ry}_{w_0}}\gategroup[2,steps=3,style={dashed,rounded
            corners,fill=red!20, inner
            sep=0.5pt},background, label style={label position=below, anchor=north,yshift=-0.2cm}]{Parametric Circuit} & \ctrl{1} & \gate{\gateg{Ry}_{w_2}} & \meter{}\\
        \lstick{$q_1: \ket{0}$} & \gate{\gateg{Rx}[x_1]} &  \gate{\gateg{Ry}_{w_1}} & \targ{} & \gate{\gateg{Ry}_{w_3}} & \ground{}
    \end{quantikz}}
    \caption{A Variational Quantum Circuit. The input $x_0,x_1$ are encoded in a quantum state by $\gateg{Rx}[x_0]$ and $\gateg{Rx}[x_1]$. The weights $w_i$ are the trainable parameters optimized during the training. The final measurement provides the values required for classical optimization.}\label{fig: VQC_example}
\end{figure}
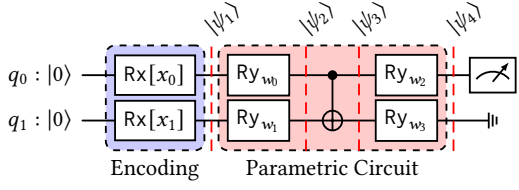
%
%
As a concrete instantiation, consider some classical data represented by the vector $[x_0,x_1]^T = [6.0, 2.7]^T$. The encoding phase applies $\gateg{Rx}[6.0]$  to $q_0$ and $\gateg{Rx}[2.7]$ to $q_1$, with $\ket{q_0q_1}=\ket{00}$, obtaining:
\begin{equation}
    \begin{aligned}\label{eq: psi1}
    \ket{\psi_1} = (\gateg{Rx}^{\tinyd{q_1}}[2.7] \cdot \gateg{Rx}^{\tinyd{q_0}}[6.0])\ket{00} =& - 0.22\ket{00} - 0.03 i\ket{01} + 0.97\ii\ket{10} - 0.14\ket{11}
    \end{aligned}
\end{equation}

On this quantum state, we can now apply the parametric quantum circuit, which, in our example, is represented by the \textit{ansatz} in \autoref{fig: VQC_example} with trainable parameters $w_0,\dots,w_3$.
As with Neural Networks, during the training phase, the parameters $w_i$ are optimized to minimize the classification error.
To distinguish the rotation gates used in the input encoding phase from those in the variational part, we will denote the latter by $\gateg{R}_w$, where $w$ is the weight optimized during training. When verifying the circuit, this is a fixed constant, as verification is performed after the training phase. 
In our example, the training phase produces $w_0 = 0.99$, $w_1 = -0.50$, $w_2 = 3.27$ and $w_3 = -0.69$.
Thus, by first applying to $\ket{\psi_1}$ the gates $\gateg{Ry}^{\tinyd{q_1}}_{0.99}$, $\gateg{Ry}^{\tinyd{q_0}}_{-0.50}$, and then 
 $\gateg{CX}^{\tinyd{q_0,q_1}}$ followed by  $\gateg{Ry}^{\tinyd{q_0}}_{3.27}$ and $\gateg{Ry}^{\tinyd{q_1}}_{ -0.69}$, we  obtain the final state $\ket{\psi_4}$:
\begin{gather}\label{eq: psi4}
    \ket{\psi_4} \!\!= (0.14 - 0.49\ii)\ket{00} - (0.11 - 0.46\ii)\ket{01} + (0.08 + 0.03\ii)\ket{10} + (0.17 + 0.70\ii)\ket{11}
\end{gather}
The last step is the measurement of the qubits in the computational basis.
We have four possible outcomes: $00$, $01$, $10$, and $11$.
By computing the squared modulus\footnote{Given a complex number $z = a + ib$, the squared modulus $|z|^2 = a^2 +b^2$.} of the amplitude of each state, we obtain the distribution $(0.26, 0.21, 0.01, 0.52)$ over the possible outcomes $\{00, 01, 10, 11\}$.
Based on this, a classification can be made by looking at the result of $q_0$.
Since $q_0$ is the least significant qubit, the probability of measuring $q_0$ in state $1$ corresponds to measuring $01$ or $11$ and can be computed as $0.21 + 0.52 = 0.73$.
Similarly, the probability of measuring it in state $0$ (i.e., measuring $00$ or $10$) is $0.26+0.01 = 0.27$.
This means that we can classify the input $[x_0,x_1]^T = [6.0, 2.7]^T$ as class $1$ with a probability of $0.73$.

In the training of a VQC, the measurement outcome of the qubit is typically used to compute a cost function, which is then optimized by using classical optimization tools to update the parameters of the circuits. 
For the verification task, which is the purpose of this paper, we can concentrate on trained VQC, where the classical optimization phase has already taken place.

\section{A language for Variational Quantum Circuits}\label{sec: vqc lan}
In order to formally verify a VQC, we introduce a language, which we call $\lang$, defining the structure and the behavior of VQCs. 
For this language, we first define a syntax (establishing the structure of a VQC), and then we give a semantics, inductively defined on the constructs of the syntax, by using a function-based representation of quantum states.

\subsection{Syntax}
Let $q, q_i, q_j$ denote the VQC qubits and $x$ a classical input variable.
We define $\lang$ as the language generated by the grammar:
\begin{equation*}
    \begin{aligned}
      &\Up := \code{Rx}^{\mbox{\tiny $q$}}[x] \mid \code{Ry}^{\tinyd{q}}[x] \mid \code{Rz}^{\tinyd{q}}[x] && \text{(Encoding operations)} \\ 
      &\U := \code{Rx}_{\tinyd{w}}^{\tinyd{q}} \mid \code{Ry}_{\tinyd{w}}^{\tinyd{q}} \mid \code{Rz}_{\tinyd{w}}^{\tinyd{q}} \mid \CN^{q_i, q_j} && \text{(Parametric operations)}\\
      &\nonter{s} := \Up \mid \U \mid \nonter{s};\nonter{s}&& \text{(Quantum statements)}\\
      &\qnn := \nonter{s};\M && \text{(VQC)}\\
    \end{aligned}
\end{equation*}

A VQC $\qnn$ is any element in the language $\lang$ 
consisting of a quantum circuit $\nonter{s}$, composed by quantum statements, followed by a final measurement operation $\M$. 
A quantum statement is a sequential composition of parametric and encoding operations.
$\Up$ represents the encoding operations, used to embed input values into quantum states. 
Each operation $\Up$ represents a rotation around one of the three axes of the Bloch sphere, with an angle determined by the input variable $x$ and applied to qubit $q$. 
$\U$ defines the operations used in the parametric (trainable) part of the circuit. 
It includes single-qubit rotations and the controlled-NOT operation $\CN^{q_i, q_j}$, where $q_i$ is the control and $q_j$ the target qubit, with $q_i \neq q_j$. 
We syntactically distinguish between $\U$ and $\Up$ to reflect their different roles within the circuit: Encoding operations depend on input variables, and are used to embed input data into quantum states; Parametric operations are parameterized by constants $w$, representing the trainable weights optimized during the training phase.
For instance, the VQC in \autoref{fig: VQC_example} can be written as: 
$$\code{Rx}^{q_0}[x_0]; \code{Rx}^{q_1}[x_1]; \code{Ry}_{w_0}^{q_0}; \code{Ry}_{w_1}^{q_1}; \CN^{q_0, q_1}; \code{Ry}_{w_2}^{q_0}; \code{Ry}_{w_3}^{q_1};\M$$
We recall that, even if we are considering only the control-not and one-qubit rotation gates, we can represent in $\lang$ any possible quantum circuit~\cite{barenco95elementary}.

\subsection{Semantics}\label{sec:semdom}
To define a semantics for $\lang$, we first need to specify the {\em domain} of the denotations of data and then a function modeling how the computation transforms such data. 
\subsubsection{Semantic Data Domains}
Let \(\varset\) represent the set of qubits in a VQC $\qnn\in\lang$. 
Each qubit \(q \in \varset\) is associated with a two-dimensional Hilbert space \(\Hilb_q\). 
Consequently, the global vector space of the program is denoted as \(\ok{\SpaceP = \bigotimes_{q \in \varset} \Hilb_q}\). The program state is then represented as a vector \(\ket{\psi} \in \SpaceP\).
In this work, we use an alternative representation rather than the standard vector formalism, which is more suitable for defining the abstract domain.  
Since a quantum state corresponds to a superposition of basis states, each one associated with a complex amplitude, we represent quantum states $\psi$ as mappings from the set $\bases$ of standard basis vectors of \(\ok{\SpaceP}\), to complex amplitudes in $\bC$.  
Formally, a state $\ok{\psi \in \DomP \defi \bases \to \bC}$ can be represented as $\psi = \lambda e \in \bases.\:z$ ($z\in\bC$),
where the normalization condition $\ok{\sum_{e \in \bases} |\psi(e)|^2 = 1}$ must hold.

%
This functional representation is equivalent to the standard vector representation. 
Throughout the paper, we use $\psi$ to denote the state as a mapping and $\ket{\psi}$ to refer to its equivalent vector representation.
For instance, consider the state $\ket{\psi_1}$ defined in \autoref{eq: psi1}. Since this is a two-qubit state, the basis set in mapping notation is $\bases = \{00, 01, 10, 11\}$. Then, the corresponding mapping is: $\psi_1 = \{00 \mapsto -0.22,\ 01 \mapsto -0.03\ii,\ 10 \mapsto +0.97\ii,\ 11 \mapsto -0.14\}$.


States are not the only data on which a VQC works. Indeed, we must also consider the input for the variables $x$. In particular, we need to assign them real values—that is, we require an \emph{input environment}.
Given $\qnn\in\lang$, let $\parset$ denote the set of classical input variables.  
We define an input environment as a mapping $\ok{\Stenv \in \DomSenv \defi \parset \to \bR}$, associating each variable in $\parset$ with a real-valued input in $\bR$.  
For instance, in the example described in \autoref{sec: vqc}, the input environment is given by: $\Stenv = \{x_0 \mapsto 6.0,\ x_1 \mapsto 2.7\}$.

\subsubsection{Quantum Statements Semantics}
We can define the semantics of a quantum statement $\nonter{s}$ as a function $\sSem[]{\nonter{s}}: \DomSenv \to (\DomP \to \DomP)$, which, given an input environment, transforms an initial state $\psi \in \DomP$ into a resulting quantum state in $\DomP$.
%
\paragraph{Parametric Operations $\,\U$}
The semantics of every parametric operation corresponds to a matrix that describes the effect of this operation (whether a single-qubit rotation or a controlled-Not) on the full Hilbert space of the program. 
For example, consider the parametric operation \(\code{Ry}_{\tinyd{w}}^{\tinyd{q}}\). 
\(\code{Ry}_{\tinyd{w}}\) corresponds to a \(2 \times 2\) unitary matrix acting on a single qubit. However, to represent the semantics of the rotation in the full program space, we must define a unitary matrix of size \(2^{|\varset|} \times 2^{|\varset|}\) for \(\code{Ry}_{\tinyd{w}}^{\tinyd{q}}\).
This larger matrix models the action of \(\code{Ry}_{\tinyd{w}}\) on qubit \(q\) while acting as the identity on all other qubits in \(\varset \smallsetminus \{q\}\). Concretely, this is achieved by taking the tensor product of the \(2 \times 2\) matrix representing \(\code{Ry}_{\tinyd{w}}\) with identity matrices corresponding to the remaining qubits.

Given an operation in $\U$, we abuse notation considering as \(\ok{\U\in\bC^{2^{|\varset|} \times 2^{|\varset|}}}\) the matrix that describes the effect of the operation in the full Hilbert space of the program. 
Given this matrix, each entry $\U_{\tinyd{e_1, e_2}}\in \bC$\footnote{We recall that the bit strings $e_1$ and $e_2$ that represent standard basis elements can be used as indices to the rows and columns of the matrix $\U$, respectively.} denotes the amplitude of transitioning from basis state $e_2$ to $e_1$, and applying $\U$ to a quantum state $\ket{\psi} \in \bC^{|\bases|}$ corresponds to the matrix-vector product $\U \cdot \ket{\psi}$. 
To define the semantics of this transformation in our setting, we must express matrix multiplication using the mapping-based formalism employed to represent abstract states.
Given an operation in $\U$, its semantics is defined as a function $\sSem[]{\U}: \DomSenv \to (\DomP \to \DomP)$, defined as:
\begin{equation}\label{eq: U sem}
    \sSem{\U} \defi \lambda \psi. \left(\lambda e \in \bases.\, \sum_{e' \in \bases} \U_{\tinyd{e,e'}} \cdot \psi(e') \right)
\end{equation}
Since $\U$ is not an encoding operation, its semantics is independent of the input environment $\sigma$.
For instance, consider a one qubit quantum state $\ket{\psi}_{q} = [a, b]^T$ (equivalently represented as $\psi_{q}=\{0\mapsto a,1\mapsto b\}$),  and the operation described by the matrix $\code{G} = \begin{bmatrix} u_{0,0} & u_{0,1}\\ u_{1,0} & u_{1,1} \end{bmatrix}$.
In vector formalism the application of $\code{G}$ to $\ket{\psi}$ is given by: $\code{G}\cdot\ket{\psi} = [u_{0,0}a + u_{0,1}b, u_{1,0}a + u_{1,1}b]^T$.
that correspond, in the mapping notation to: $\csSem[]{\code{G}}(\psi)= \{0\mapsto u_{0,0}a + u_{0,1}b, 1\mapsto u_{1,0}a + u_{1,1}b\} = \{0\mapsto u_{0,0}\cdot\psi(0) + u_{0,1}\cdot\psi(1), 1\mapsto u_{1,0}\cdot\psi(0) + u_{1,1}\cdot\psi(1)\}$. 

\paragraph{Encoding Operations $\Up$}
We can adopt an approach similar to the one used for parametric operations. 
$\Up$ represents single-qubit rotations by an angle depending on the value of $x$.
This means that every rotations are describe by a matrix over functions, $\fun^{2 \times 2}$, where $\fun = \{ f[x] \mid f: \mathbb{R} \to \mathbb{C},\ x \in \parset \}$ is a set of complex-valued functions over the variable $x$.
The matrix $\gateg{Rx}[x]$ in \autoref{eq: Rx} is an example of such a rotation in $\Up$. 
Abusing notation, by $\Up$ we refer to the matrix modeling the encoding operation acting on the entire system: the matrix of the specific rotation on a qubit in tensor with the identity on the rest. 
Accordingly, $\ok{\Up \in \fun^{2^{|\varset|} \times 2^{|\varset|}}}$.

In order to define the semantics of $\Up$, we first need to evaluate the functions $\func[x] \in \fun$ with respect to an input environment $\Stenv \in \DomSenv$.  
We define the interpretation function $\grassto{\func[x]} : \DomSenv \to \mathbb{C}$ as: $\grassto{\func[x]}_{\Stenv} = \func(\Stenv(x)) \in \mathbb{C}$, that is, we evaluate the function $\func$ on the input value $\Stenv(x)$.
Given an encoding operation $\Up$, each matrix entry $\Up_{\tinyd{e_1,e_2}}$ is a function in $\fun$, mapping an input value to a complex coefficient. 
The semantics of $\Up$ is a function $\sSem[]{\Up}: \DomSenv \to (\DomP \to \DomP)$ defined as:
\begin{equation}\label{eq: Up-sem}
    \sSem{\Up} \defi \lambda \psi.\:(\lambda e \in \bases.\:\sum_{e' \in \bases} \grassto{\Up_{\tinyd{e,e'}}}_\Stenv \cdot \psi(e'))
\end{equation}


\paragraph{Sequential Composition}
The last operator for building a VQC is sequential composition of statements $\nonter{s}_1;\, \nonter{s}_2$. 
In this case, the semantics is trivially the functional composition of the semantics:
\[
\sSem{\nonter{s}_1;\, \nonter{s}_2} \defi \sSem{\nonter{s}_2} \comp \sSem{\nonter{s}_1}
\]

\subsubsection{The Semantics of a VQC}
Since a VQC $\qnn\in\lang$ is built as $\qnn:= \nonter{s};\M$, its semantics requires a denotation for the final measurement operator. 
This can be defined as a function that, given a quantum state, returns a probability distribution over the set of possible results.
In our setting, we define $\ok{\DomR \defi \bases \to \bR}$ (with metavariable $\dist$), the set of all possible probability distributions over the bases $\bases$.
The measurement semantics can be expressed (abusing notation) as a function $\sSem[]{\M}:\DomP \to \DomR$, formally defined as:
\begin{equation}\label{eq: Msem}
    \sSem[]{\M}\defi\lambda \psi.\:(\lambda e \in \bases.\:|\psi(e)|^2).
\end{equation}
For instance, if we consider the state in \autoref{eq: psi4}, that in the mapping representation is $\psi_{4} = \{00\mapsto (0.14 - 0.49\ii), 01\mapsto -(0.11 - 0.46\ii), 10\mapsto (0.08 + 0.03\ii), 11\mapsto (0.17 + 0.70\ii)\}$.
The result of the application measurement operators is: 
$$
\sSem[]{\M}(\psi_{4}) = \{00\mapsto 0.26,01\mapsto 0.21,10\mapsto 0.01,11\mapsto 0.52\}\in\DomR[\tinyd{q_0,q_1}]
$$

Finally, the concrete semantics of the entire $\qnn= \nonter{s};\M$, is the composition of the semantics of the measurement applied to the result of operator $\nonter{s}$ semantics, obtaining (again abusing notation) a function $\sSem[]{\qnn}\!\!:\DomSenv\to(\DomP\to\DomR)$ defined as:
\[
\sSem{\qnn} \defi \lambda \psi.\: \sSem[]{\M}\comp\sSem{\nonter{s}}(\psi).
\]

When verifying a VQC, we execute the circuit on a set of input values, which corresponds to a specific input and all possible values obtained by perturbing it. 
In the concrete domain, this would involve executing the circuit for every possible input.
To formalize this, we define the concrete semantics of the circuit in a collecting style, i.e., as a function $\csSem[]{\qnn}:\wp(\DomSenv)\to (\wp(\DomP) \to \wp(\DomR))$. 
We abuse notation by using the same semantic notation for the additive lift of $\sSem[]{\cdot}$ to sets of environments and states, i.e., $\ok{\csSem{\qnn}(\Psi)\defi\sset{\sSem{\qnn}(\psi)}{\sigma\in\Sigma,\psi\in\Psi}}$, where $\Sigma\in\wp(\DomSenv)$ and $\Psi\in\wp(\DomP)$ represent a set of initial states of the VQC.

\section{Abstracting $\lang$ Semantics}\label{sec: abstract}
In many fields of computation, the concrete semantics is usually an intractable model of computation due to computability or computational complexity reasons. In this section, we follow the idea widely exploited both in programming languages and in NN~\cite{cousot1979systematic,Ai2}, consisting of executing the computational system on {\em approximated data}, e.g., intervals of numerical values.
In order to approximate the computation, we first need to abstract the data. 
We observe that inputs (environments) and outputs (distributions) are mappings to real values, while the VQC works on states, which are mappings to complex values. 
This means that we have to split the abstraction process, namely, we have to introduce the real intervals abstraction for environments and distributions, and the complex interval abstraction for the states.

\subsection{Abstracting Environments (inputs) and Distributions (outputs)}

Let us first extend the standard abstract domain of intervals to real values, with a slight change consisting of considering only bounded intervals.

\subsubsection{Reals Interval Domain}

Let us consider the domain of closed~\cite{cousot_2021_dynamic} and bounded~\cite{Pulina2010} intervals on $\bR$, i.e., $\bRI\defi\sset{\interv{r_l}{r_u}}{r_l,r_u\in\bR}\cup\{\varnothing,\bR\}$. As it happens for integer intervals, also this domain can be characterized as an abstraction of the powerset of reals $\tuple{\wp(\bR),\subseteq}$~\cite{cousot_2021_dynamic}.
We define two functions $\alphaR:\wp(\bR)\to\bRI$ and $\gammaR:\bRI\to\wp(\bR)$ as: Let $I\in\bRI$
\begin{align*}
\gammaR(I)\defi&\left \{
\begin{array}{ll}
     \sset{r}{r_l\leq r\leq r_u}& \mbox{if $I=\interv{r_l}{r_u}$} \\
     I& \mbox{if}\ I=\varnothing\ \mbox{or}\ I=\bR
\end{array}
\right.
\quad
\alphaR(X)\defi&\bigcap\sset{I\in\bRI}{X\subseteq\gammaR(I)}
\end{align*}
Note that the definition of $\alphaR$ is a good definition since closed intervals on $\bR$ are closed sets in the corresponding topological space, and the intersection of any family of closed sets is a closed set, with $\varnothing$ and $\bR$ closed sets.
In general, let $X\in\wp(\bR)$,  $\alphaR(X)=[\min(X),\max(X)]$ if both $\min(X)$ and $\max(X)$ exist in $X$;  If at least one of them does not exist since $X$ has an {\em open} bound, e.g., $X=\sset{r}{2<r\leq 10}$, then the abstraction closes it, namely $\alphaR(X)=[2,10]$; If, at least one bound not exist since the values in $X$ grow to infinite (or decrease to negative infinite), e.g., $X=\sset{r\in\bR}{r\geq 10}$, then $\alphaR(X)=\bR$; If $X=\varnothing$ then the abstraction is $\varnothing$. When $\alphaR(X)=[l,u]$ we abuse notation writing $l=\min(X)$ and $u=\max(X)$~\cite{Pulina2010}.

The following results come from the fact that $\gammaR$ is a co-additive function over $\bRI$, and therefore, by construction, $\alphaR$ is its adjoint function forming a Galois Insertion.


\begin{proposition}
    $\tuple{\wp(\bR),\subseteq}\galoiS{\alphaR}{\gammaR}\tuple{\bRI,\leq_{\tinyd{\bRI}}}$, where $\leq_{\tinyd{\bRI}}$ is standard inclusion between intervals.
\end{proposition}

In~\cite{cousot_2021_dynamic,intervalbook} we can find the whole arithmetic on real intervals. 
Here we will use only multiplication (extended to matrices) and sum, that we denote, abusing notation, as $+$ and $\cdot$.

\subsubsection{Abstract Environments}
Environments map input variables to real values. 
In order to abstract the VQC computation to intervals, we have to define {\em abstract environments} $\abenv$ mapping the input variables of $\qnn\in\lang$ (i.e., those in $\parset$) to the {\em intervals} of values in $\bRI$ representing the range of possible values for that variable. This abstraction is performed on the collection of reals associated with each variable (due to a set of possible input environments). 
For example, if a variable $x\in\parset$ can take values between $0$ and $r\in\bR$, the abstract environment would map $x$ to the interval $[0, r]$.  
Formally, $\ok{\abenv \in \aDomSenv \defi \parset \to \bRI}$ can be defined by abstract interpretation on the power domain of environments. 
Let $\ok{\alphaE : \wp(\DomSenv) \to \aDomSenv}$ defined as follows: Let $\Sigma\in\wp(\DomSenv)$
\begin{align*}
    \alphaE(\Sigma) &\defi \lambda x \in \parset.\:\alphaR\big(\sset{\sigma(x)}{\sigma\in\Sigma}\big),
\end{align*}
This function abstracts a set of inputs (i.e., a set of static environments) into a single abstract environment $\abenv$.
The concretization is the standard adjoint $\gammaE:\aDomSenv\to\wp(\DomSenv)$, namely $\forall \abenv\in\aDomSenv$ we have $\ok{\gammaE(\abenv)\defi\sset{\Stenv\in\DomSenv}{\forall x\in\parset.\:\Stenv(x)\in\gammaR(\abenv(x))}}$.
\begin{proposition}
    $\tuple{ \wp(\DomSenv),\subseteq}\galoiS{\alphaE}{\gammaE}\tuple{\aDomSenv,\dot{\leq}_{\tinyd{\bRI}}}$\footnote{$\dot{\leq}_{\tinyd{\bRI}}$ is the point-wise ordering induced by $\leq_{\tinyd{\bRI}}$ on $\aDomSenv$, i.e., $\forall\abenv_1,\abenv_2\in\aDomSenv$ we have $\abenv_1\dot{\leq}_{\tinyd{\bRI}}\abenv_2$ iff $\forall x\in\parset.\:\abenv_1(x)\leq_{\tinyd{\bRI}}\abenv_2(x)$.}.
\end{proposition}

\subsubsection{Abstract Distributions} The output of the VQC semantics is a probability distribution, which must also be abstracted in order to return intervals of probabilities when measuring an abstract state. 
Formally, let us define $\ok{\aDomR\defi \bases \to \bRI}$, the abstract domain of probability distributions on intervals, point-wise ordered by $\dot{\leq}_{\tinyd{\bRI}}$. We can define 
the abstraction function $\ok{\alphaD: \wp(\DomR) \to \aDomR}$,
between sets of probability distributions $\tuple{\wp(\DomR),\subseteq}$ and abstract interval distributions $\ok{\tuple{\aDomR,\dot{\leq}_{\tinyd{\bRI}}}}$. 
Let $\Delta\in \wp(\DomR)$,
\[
\alphaD(\Delta) \defi \lambda e \in \bases.\:\alphaR(\sset{\dist(e)}{\dist\in \Delta})
\]
The concretisation function $\gammaD:\aDomR\to\wp(\DomR)$ is the standard adjoint, i.e., $\forall \absdist\in\aDomR.\:\gammaD(\absdist)\defi\sset{\dist\in\DomR}{\forall e\in\bases.\:\dist(e)\in\gammaR(\absdist(e))}$.

For instance, the mapping $\absdist = \{00 \mapsto  [0.126, 0.50], 01 \mapsto  [0.116, 0.460], 10 \mapsto [0, 0.140], 11 \mapsto  [0.291, 0.761]\}$ in \autoref{ex:absdist}, express the abstract state where  we can obtain $00$ with probability (w.p.) $[0.126, 0.50]$, $01$ w.p. $[0.116, 0.460]$, $10$ w.p. $[0, 0.140]$, and $11$ w.p. $[0.291, 0.761]$.

\subsection{Abstracting Quantum States}
In order to abstract states, mapping quantum variables to complex values, into the interval domain, we need to extend the abstract domain of intervals to complex values, i.e., to pairs of real intervals. 

\subsubsection{Complex Interval Domain}

Let us consider the set of complex numbers $\bC$.
Given a generic complex number $z=r_1 + \ii r_2\in \bC$, we denote by $\Re(z) = r_1$ its real part, and by $\Im(z) = r_2$ its imaginary part.
Analogously, given a $Z \in \wp(\bC)$, $\ok{\Re(Z) \defi \sset{\Re(z)}{z \in Z }}$ and $\ok{\Im(Z) \defi \sset{\Im(z)}{z \in Z}}$.
\\
Now, we can abstract complex numbers to intervals simply by keeping the information that $z$ is determined by the pairs of real values $r_1$ and $r_2$. Following this idea, we define the abstract domain of complex intervals as $\ok{\bCI \defi \bRI \times \bRI}$. Hence, we can trivially extend the interval abstraction and concretization functions to complex numbers as: Let $Z\in\wp(\bC)$ and $R, I\in\bRI$
\begin{gather*}
        \alphaC : \wp(\bC) \to \bCI,\ \alphaC(Z) \defi \tuple{\alphaR (\Re(Z)), \alphaR (\Im(Z))}\\
        \gammaC : \bCI \to \wp(\bC),\ \gammaC(R,I) \defi \sset{z \in \bC}{\Re(z) \in \gammaR(R),\: \Im(z) \in \gammaR(I)}
\end{gather*}

\begin{proposition}
    $\tuple{\wp(\bC),\subseteq}\galoiS{\alphaC}{\gammaC}\tuple{\bCI,\leq_{\tinyd{\bCI}}}$, where $\forall \tuple{R_1,I_1}, \tuple{R_2,I_2}\in\bCI$ we have $\tuple{R_1,I_1}\leq_{\tinyd{\bCI}}\tuple{R_2,I_2}$ iff $R_1\leq_{\tinyd{\bRI}} R_2$ and $I_1\leq_{\tinyd{\bRI}} I_2$.
\end{proposition}

We abuse notation by using the same operator introduced on real intervals also on complex intervals, with the only difference of being point-wise applied to the elements of the pair, i.e., $\tuple{[a_1,b_1],[c_1,d_1]}+\tuple{[a_2,b_2],[c_2,d_2]}=\tuple{[a_1,b_1]+[a_2,b_2],[c_1,d_1]+[c_2,d_2]}$, analogous for $\cdot$~\cite{intervalbook,gargantini1971circular}. 


\subsubsection{Abstract States}
%

Similarly to what we did for the abstract environments, we now need to map each standard basis vector to an interval of complex values, in order to abstract states. 
We define the set of abstract states $\aDomP= \bases \to \bCI$ mapping each standard basis $e\in\bases$ to a complex abstract interval $\tuple{\interv{a}{a'},\interv{b}{b'}}\in\bCI$.
This domain $\aDomP$ can be obtained as the abstraction function $\ok{\alphaP : \wp(\DomP)\to \aDomP}$:  
Let $\Psi\in\wp(\DomP)$
\[
\alphaP(\Psi) \defi \lambda e \in \bases.\:\alphaC(\sset{\psi(e)}{\psi\in \Psi})
\]
The concretisation function $\gammaP:\aDomP\to\wp(\DomP)$ is the standard adjoint, i.e., is defined as $\forall\abste\in\aDomP.\:\gammaP(\abste)\defi\sset{\psi\in\DomP}{\forall e\in\bases.\psi(e)\in\gammaC(\abste(e))}$.
\begin{proposition}
    $\tuple{\wp(\DomP),\subseteq}\galoiS{\alphaP}{\gammaP}\tuple{\aDomP,\dot{\leq}_{\tinyd{\bCI}}}$ ($\dot{\leq}_{\tinyd{\bCI}}$ is the point-wise ordering 
    induced by $\leq_{\tinyd{\bCI}}$).
\end{proposition}


An example of an abstract state is in \autoref{ex:absstate}.
Here, $\abste_4(00) = \tuple{[-0.159, 0.433], [-0.559, -0.355]}$, meaning that the real part of the amplitude associated with $00$ ranges between -0.159 and 0.433 while the imaginary part ranges between -0.559 and -0.355.


 



\subsection{Abstract Semantics of $\lang$}\label{sec: abssem}

Starting from an abstract environment $\abenv$ and an abstract state $\abste$, both working on intervals, we can define the abstract semantics of a VQC $\qnn\in\lang$ executing $\qnn$ on intervals.

We first have to define an abstract evaluation of functions $\func[x]\in\fun$ w.r.t.\ the abstract environment, and therefore we have to compute $\func$ on intervals.
Let $\abenv\in\aDomSenv$, we define $\agrassto{\func[x]}:\aDomSenv\to\bCI$ as an approximation of $\func$ on intervals: 
$\ok{\agrassto{\func[x]}_{\tinyd{\abenv}}\defi\alphaC\comp \func\comp\gammaC(\abenv(x))\in\bCI.}$
This is the best correct approximation of (the additive lift to sets of) $\func$ on $\bCI$.

Given $\qnn=\nonter{s};\M\in\lang$, in order to define its abstract semantics, we need to define separately the abstract semantics of quantum statements $\nonter{s}$ and of the measurement operator $\M$.
\begin{definition}[Abstract Quantum Statements Semantics]\label{def: abssem}
    For each statement $\nonter{s}$, the abstract semantics is the function $\abSem{\nonter{s}}{}:\aDomSenv\to (\aDomP\to\aDomP)$ is defined as follows: Let $\abste\in\aDomP$
    \begin{enumerate}
        \item $\abSem{\U}{\abenv} \defi \lambda \abste.  \left(\lambda e \in \bases . \sum_{e' \in \bases} (\U_{\tinyd{e,e'}}\cdot\abste(e'))  \right)$
        \item $\abSem{\Up}{\abenv} \defi \lambda \abste.\:(\lambda e \in \bases.\: \sum_{e'\in \bases} \agrassto{\Up_{\tinyd{e,e'}}}_{\tinyd{\abenv}}\cdot
        \abste(e')) $;
        \item $\abSem{\nonter{s}_1;\, \nonter{s}_2}{\abenv} \defi \abSem{\nonter{s}_2}{\abenv} \circ \abSem{\nonter{s}_1}{\abenv}$,
    \end{enumerate}  
\end{definition}
%
\noindent
Since $\Up_{\tinyd{e,e'}}$ is always a sum of sine and cosine, $\agrassto{\Up_{\tinyd{e,e'}}}_{\tinyd{\abenv}}$ can be defined using the standard interval arithmetic~\cite{intervalbook}.

Then, let $\M$ be a measurement operator, the abstract semantics of $\M$ is a function $\abSem{\M}{}: \aDomP \to \aDomR$ defined as:
$$
\abSem{\M}{} \defi \lambda\abste.(\lambda e\in\bases.\:|\abste(e)|^2)
$$
where $|\cdot|^2$ is computed according to the interval arithmetic~\cite{intervalbook}.

Finally, the abstract semantics of $\qnn := \nonter{s};\M$ is a function $\abSem{\qnn}{}: \aDomSenv\to(\aDomP \to \aDomR)$ defined as:
$$
\abSem{\qnn}{\abenv} \defi \lambda \abste.\:\abSem{\M}{} \circ \abSem{\nonter{s}}{\abenv}(\abste).
$$



The following result tells us that computing a VQC on the abstract domain of intervals by using the so far defined abstract semantics $\abSem{\cdot}{}$ preserves all the concrete computations.

\begin{theorem}[Soundness]\label{th:sound}
    The abstract VQC semantics $\abSem{\cdot}{}$ is sound. Formally, $\forall\Sigma\in\wp(\DomSenv)$ and $\forall \Psi \in \wp(\DomP)$, then $\alphaD (\csSem{\qnn}(\Psi)) \dot{\leq}_{\tinyd{\bRI}} \abSem{\qnn}{\alphaE(\Sigma)}(\alphaP(\Psi))$
\end{theorem}
\begin{proof}
    The soundness comes directly from the fact that the abstract semantics corresponds to a composition of operations in the interval arithmetic, and these operations are all sound.
\end{proof}

\subsection{An Example}\label{sec: vqc ab example}

We consider the same VQC in \autoref{sec: vqc}.
Now, suppose we want to execute the circuit for all inputs generated by perturbing the original input \( [x_0,x_1]^T = [6.0, 2.7]^T \) by an  
\( \epsilon = 0.5 \). 
This perturbation will create a the set of environments \( \Sigma = \sset{ \sigma}{ \mid 5.5 < \sigma(x_0) < 6.5 \wedge 2.2 < \sigma(x_1) < 3.2} \). 
To do this, we first define the abstract environment as \( \alphaE(\Sigma) = \abenv: \{ x_0 \to [5.5, 6.5], x_1 \to [2.2, 3.2] \} \).
Now we can execute our circuit abstractly.
We start from the initial abstract state $\init{\abste} =\{00 \to \tuple{[1,1],[0,0]}, (01,10,11) \to \tuple{[0,0],[0,0]}\}$, which represents the abstraction of the concrete initial state $\init{\psi}$ in \autoref{sec: vqc}.
The first step is to execute the abstract semantics of the encoding. 
We need to evaluate the abstract semantics of $\gateg{Rx}^{\tinyd{q_0}}[x_0]$ and $\gateg{Rx}^{\tinyd{q_1}}[x_1]$ whit $\abenv(x_0) = [5.5, 6.5]$ and $\abenv(x_1) = [2.2, 3.2]$, that corresponds to:  
\begin{equation*}
\abste_1 \!=\! \abSem{\gateg{Rx}^{\tinyd{q_1}}[x_1]}{\abenv} \circ \abSem{\gateg{Rx}^{\tinyd{q_0}}[x_0]}{\abenv}(\init{\abste}) \!=\!
\begin{Bmatrix}
00 \mapsto \tuple{[-0.454, 0.029], [0, 0]}\!\!\! & \!\!\!01 \mapsto \tuple{[0, 0], [-0.173, 0.049]}  \\
10 \mapsto \tuple{[0, 0], [0.824, 1.0]} & 11 \mapsto \tuple{[-0.382, 0.108], [0, 0]}
\end{Bmatrix}.
\end{equation*}

At this point, we proceed by applying the variational part of the circuit, which computes the classification.
As in the concrete case, this is obtained applying, in sequence, the following operations $\gateg{Ry}^{\tinyd{q_0}}_{-0.50}$, $\gateg{Ry}^{\tinyd{q_1}}_{0.99}$, $\gateg{CX}^{\tinyd{q_0,q_1}}$, $ \gateg{Ry}^{\tinyd{q_0}}_{3.27}$ and $\gateg{Ry}^{\tinyd{q_1}}_{-0.69}$.
The final state is thus obtained by computing:
\begin{equation}\label{ex:absstate}
    \begin{aligned}
    \abste_{\tinyd{4}} &= \abSem{\gateg{Ry}^{\tinyd{q_1}}_{-0.69}}{\abenv} \circ \abSem{\gateg{Ry}^{\tinyd{q_0}}_{3.27}}{\abenv} \circ \abSem{\gateg{CX}^{\tinyd{q_0,q_1}}}{\abenv} \circ \abSem{\gateg{Ry}^{\tinyd{q_1}}_{0.99}}{\abenv}\circ\abSem{\gateg{Ry}^{\tinyd{q_0}}_{-0.50}}{\abenv}(\abste_1) =\\
     & = \begin{Bmatrix} 
     00 \mapsto \tuple{[-0.159, 0.433], [-0.559, -0.355]} & 
     01 \mapsto \tuple{[-0.404, 0.184], [0.341, 0.544]}\\ 
     10 \mapsto \tuple{[-0.171, 0.328], [-0.074, 0.181]} & 
     11 \mapsto \tuple{[-0.086, 0.413], [0.54, 0.768]} \end{Bmatrix}.
\end{aligned}
\end{equation}

\noindent
The last step is the measurement. We compute $\abSem{\M}{}(\abste_4)$ obtaining the following abstract distribution: 
\begin{align}\label{ex:absdist}
    \absdist = \begin{Bmatrix} 00 \mapsto  [0.126, 0.50] & 01 \mapsto  [0.116, 0.460]\\ 10 \mapsto [0, 0.140] & 11 \mapsto  [0.291, 0.761] \end{Bmatrix}
\end{align}
%
As we see in \autoref{fig: VQC_example}, we consider only the measurement of $q_0$.
Thus, as in the concrete state we sum the probability intervals of $00$ and $10$, obtaining that the probability of measuring $0$ is $[0.126, 0.640]$ and we sum the probability of $01$ and $11$ obtaining that the probability of measuring $1$ ranges between $[0.407, 1,221]$.
The fact that the probability of measuring $1$ exceeds $1$ is due to the approximations of our abstraction.

\section{On the Precision of $\lang$ Abstract Semantics}
\label{sec: precision}
In the previous section, we have introduced a {\em sound} abstract semantics. 
Soundness means that the abstract computation may add imprecision, potentially over-approximating the concrete results. 
In our case, this manifests as computing intervals that are larger than those we would obtain by abstracting the result of a concrete execution.

\subsection{Sources of incompleteness}
In this section, we discuss where and how the abstract computation of a VQC loses precision with respect to the concrete one.
To better understand the source of imprecision, we appeal to the concept of \emph{completeness} \cite{cousot1977abstract,cousot1979systematic} of the abstraction.  
Analyzing completeness allows us to identify whether the loss of precision stems from the abstraction's loss of information or from the computations over the abstract domain.

\begin{lemma}
    Let $\func: \bR^n \to \bR$ be defined as $\func(\tuple{r_i}_{i=1}^n) = \sum_{i=1}^n r_i$, where $\tuple{r_i}_{i=1}^n$ denotes a tuples of $n$ reals.  
    The abstraction of $\func$ over the interval domain is not (globally) complete.
\end{lemma}
\begin{proof}
    We proceeded to find counterexamples for completeness.
    Let us recall that $\func^\sharp: \bRI^n \to \bRI$ and it is defined using the interval sum.
    Let $X \in \wp(\bR^n)$ be a set of tuples, we write $X_j$ as the set $\sset{r_j}{\forall \tuple{r_i}_{i=1}^n \in X.\; r_j \in \tuple{r_i}_{i=1}^n}$, i.e., the all set elements in position $j$ in the tuples in $R$.
    We call $\alphaRn: \wp(\bR^n) \to \bRI^n$ the abstraction from a set of tuples to an interval tuple defined as: $\alphaRn(X) = \tuple{\alphaR(X_j)}_{j=1}^n$.
    Let $X = \sset{\tuple{r_i}_{i=1}^n}{\sum_{i=1}^n r_i^2 = 1}$ a set of normalized reals tuples such that $\opint{X} = \alphaRn(X) = \tuple{[-1,1]}_{j=1}^n$.
    Thus $\func^\sharp(\opint{X}) = \sum_{i=1}^n [-1, 1] = [-n, n]$.
    On the other hand, the concrete image of $X$ through $\func$ is: $\func(X) = \sset{ \sum_{i=1}^n r_i}{\tuple{r_i}_{i=1}^n \in X}$, and by the Cauchy-Schwarz inequality and the normalization constraint: $\left| \sum_{i=1}^n r_i \right| \leq \sqrt{n} \cdot \left( \sum_{i=1}^n r_i^2 \right)^{1/2} = \sqrt{n}$.
    Hence: $\alphaR(\func(X)) \leq_{\tinyd{\bRI}} [-\sqrt{n}, \sqrt{n}] <_{\tinyd{\bRI}} [-n, n] = \func^\sharp(\alphaRn(X))$ i.e., the abstraction is not complete.
\end{proof}

From the previous lemma, we derive a more general result. The counterexample was based on tuples of normalized real numbers, i.e., $\sum r_i^2 = 1$.
As shown in \autoref{eq: U sem} and \autoref{eq: Up-sem}, the semantics of our quantum operators rely on sums of products of complex numbers. 
Since addition and multiplication over $\bC$ decompose into sum of real numbers, the semantics reduces to real-valued functions over tuples of normalized values,\footnote{Both state vectors $\psi$ and matrix coefficients are normalized, due to unitarity and state normalization constraints.} and by the previous lemma, incompleteness arises in such settings.

\begin{theorem}[Incompleteness]\label{th: incompleteness}
    The abstract semantics $\abSem{\cdot}{}$ of $\lang$ is not complete.
\end{theorem}

We show incompleteness by means of the following example.
\begin{example}\label{ex: incompletness}
    Let us consider the following circuit on one qubit $q$: $\nonter{s} := \code{Rx}^{\tinyd{q}}[x];\code{Rx}_{\sfrac{\pi}{2}}^{\tinyd{q}}$, on the set of inputs $\Sigma: \sset{\sigma}{\sigma(x) \in [\sfrac{\pi}{2} - 0.2, \sfrac{\pi}{2} + 0.2]}$.
    We start from initial state $\psi_{0} = \{0 \mapsto 1, 1 \mapsto 0\}$.
    First, we compute the abstract semantics of the first rotation gate:
    \[
     \abste_1 = \abSem{\code{Rx}^{\tinyd{q}}[x]}{\abenv}(\alphaP(\{\init{\psi}\}))
    = \{0 \mapsto ([0.51, 0.86], [0, 0]), 1 \mapsto ([0, 0], [-0.86, -0.51])\}.
    \]
    Then, applying the $\code{Rx}_{\sfrac{\pi}{2}}^{\tinyd{q}}$ rotation, we obtain:
    \[
    \abste_f = \abSem{\code{Rx}_{\sfrac{\pi}{2}}^{\tinyd{q}}}{\abenv}(\abste_1) = \{0 \mapsto ([-0.247, 0.247], [0, 0]), 1 \mapsto ([0, 0], [-1.216, -0.722])\}.
    \]
    We observe that abstraction introduces components with norms greater than one, since the imaginary component of the amplitude associated with basis state $1$ exceeds 1, clearly values that result from over-approximation noise.
    Thus this examples shows that, given $\init{\psi} = \{0 \to 1, 1 \to 0\}$, $\alphaC\circ\sSem[\Sigma]{\nonter{s}}(\{\init{\psi}\}) \dot{<}_{\tinyd{\code{CI}}} \abSem{\nonter{s}}{\alphaE(\Sigma)}(\alphaP(\{\init{\psi}\}))$, i.e. the abstract semantics is incomplete.
\end{example}
We conclude that the non-relational nature of the interval domain introduces significant over-approximation, especially in the quantum setting, where amplitudes and matrix coefficients are inherently dependent due to normalization and unitary constraints.

A visual example of how overapproximation arises in the 2D complex plane is shown in \autoref{fig: interval-2d}.  
\begin{figure}
    \centering
    \begin{subfigure}[c]{0.32\textwidth}
    \centering
    \includegraphics[trim=20cm 20cm 0em 0em, clip,width=.6\textwidth]{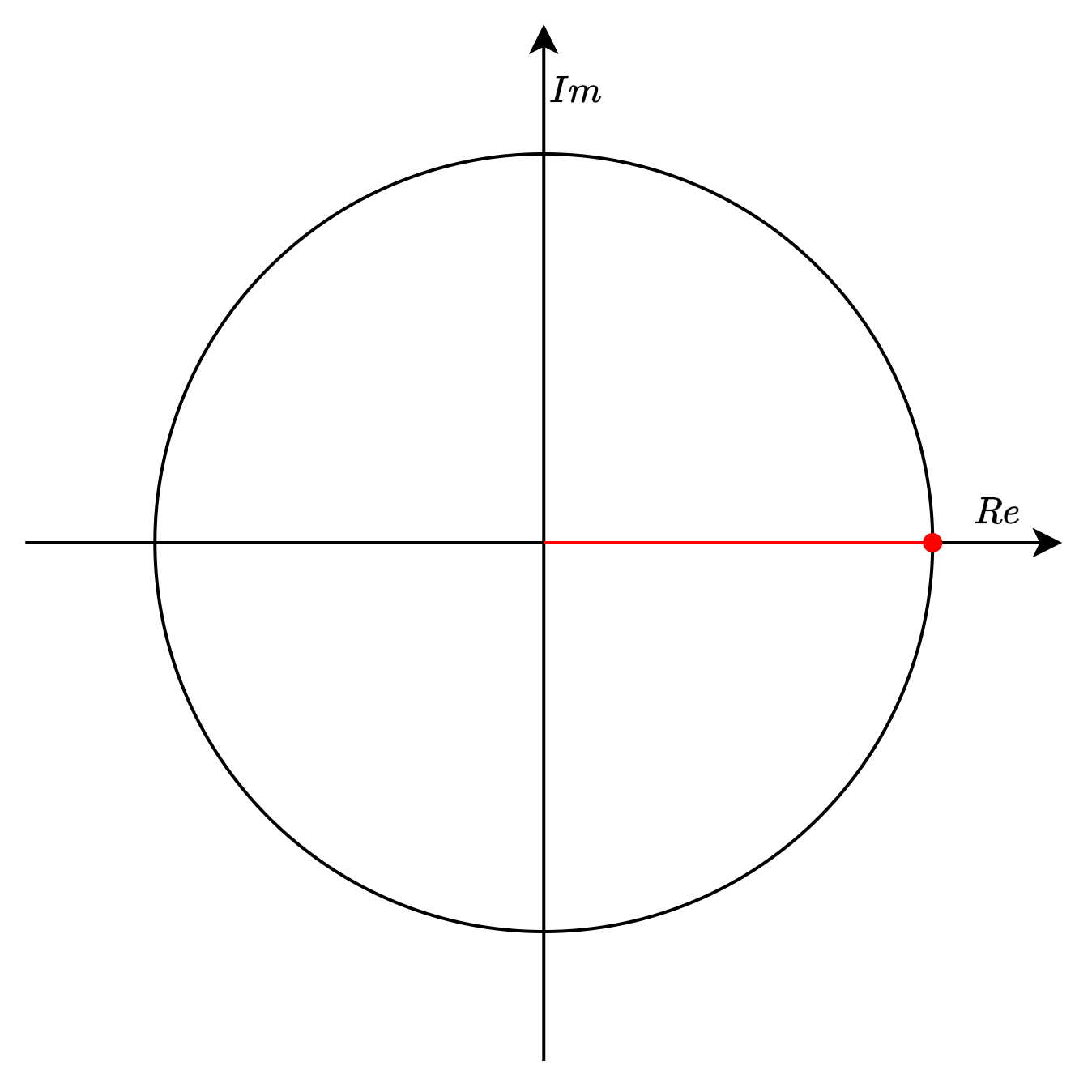}
    \caption{}\label{fig: init-state}
    \end{subfigure}
    \begin{subfigure}[c]{0.32\textwidth}
    \centering
    \includegraphics[trim=20cm 20cm 0em 0em, clip,width=.75\textwidth]{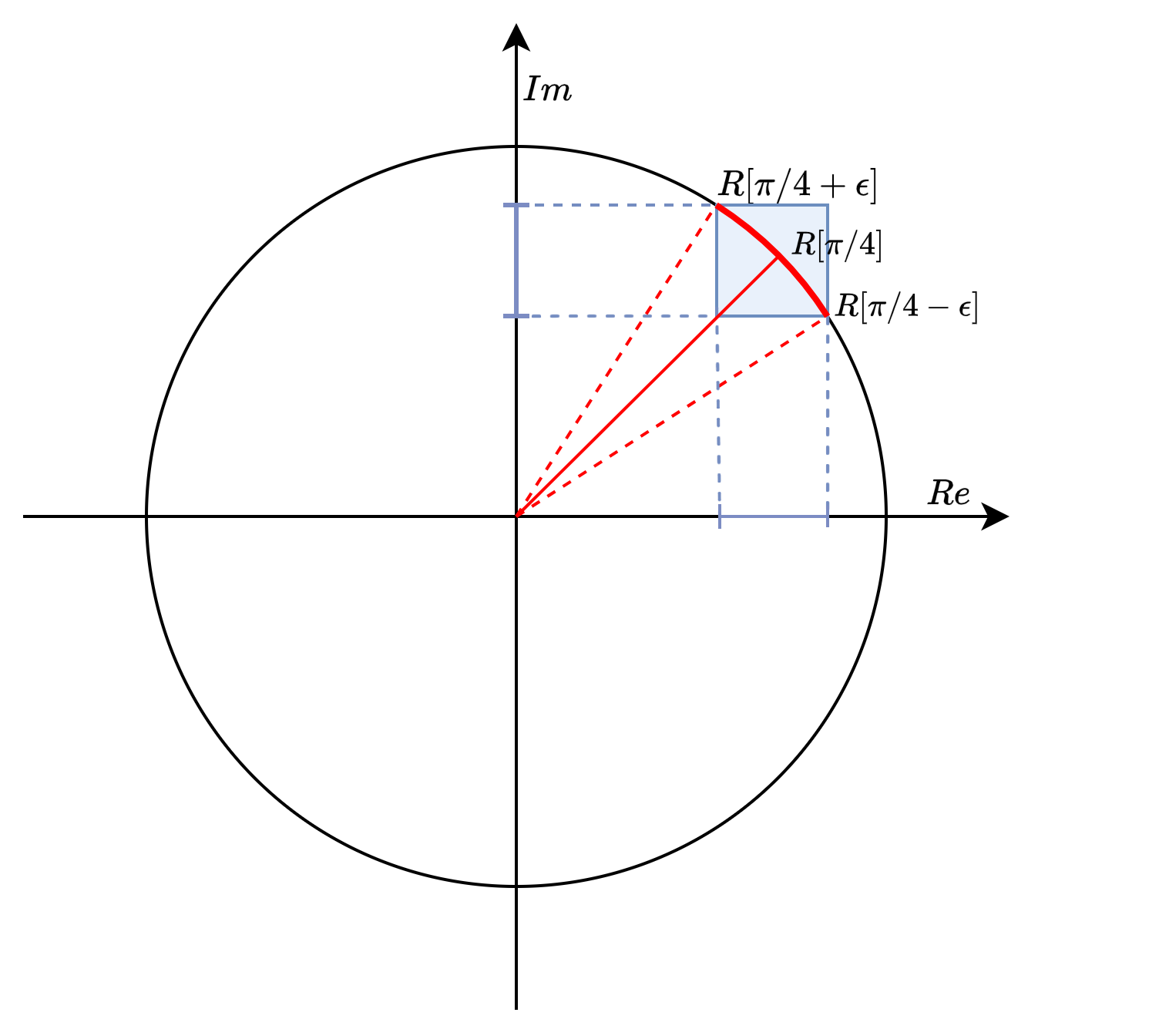}
    \caption{}\label{fig: interval-approx1}
    \end{subfigure}
    \begin{subfigure}[c]{0.32\textwidth}
    \includegraphics[trim=10cm 20cm 0em 0em, clip,width=.75\textwidth]{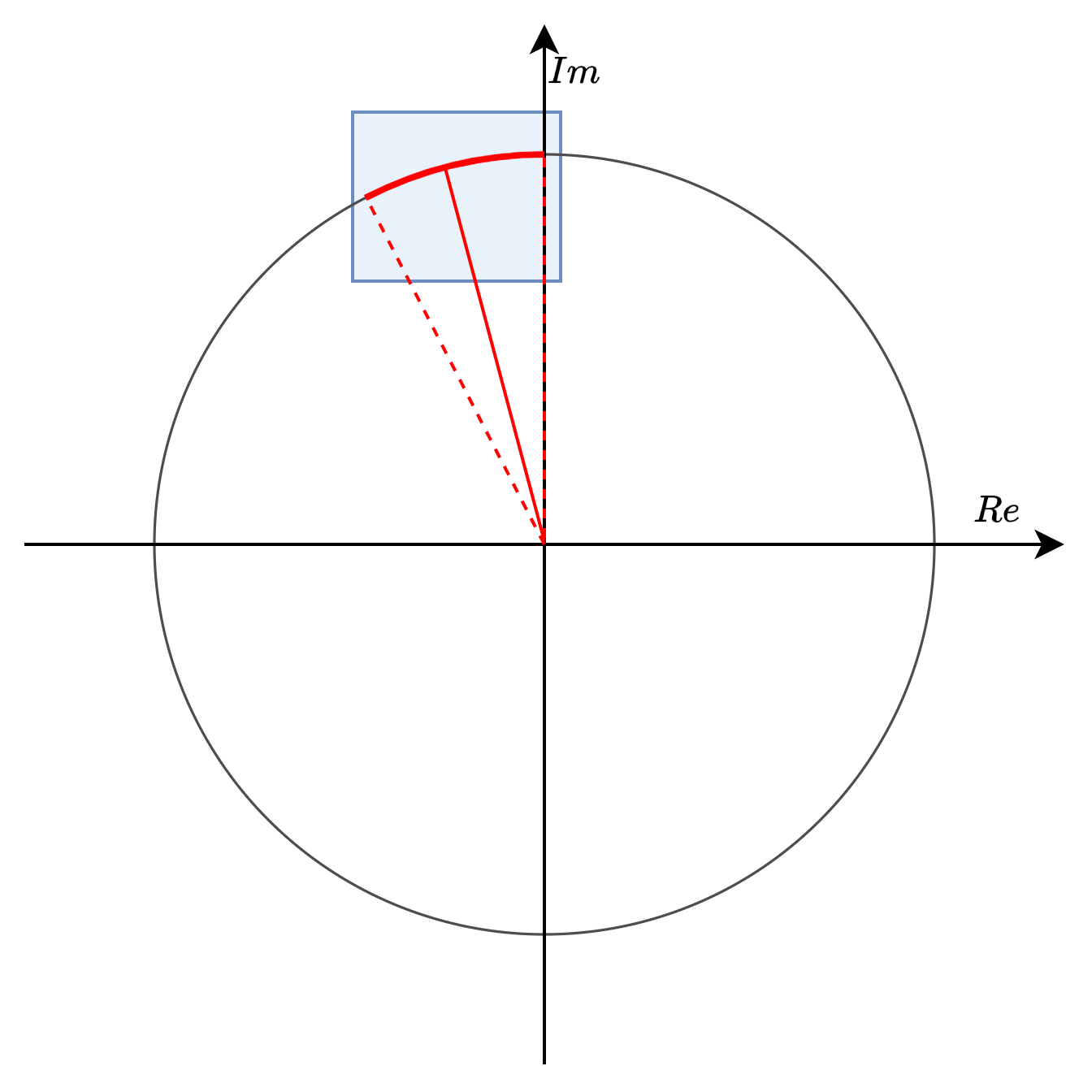}
    \caption{}\label{fig: interval-approx2}
    \end{subfigure}
    \caption{We represent a complex number as a point with real and imaginary parts on the $x$ and $y$ axes. In (a) we represent $z = \tuple{1,0}$; in (b) the results of computing the rotation $R[\theta](z)$ with $\theta \in [\frac{\pi}{4} - \epsilon,\, \frac{\pi}{4} + \epsilon]$ abstractly (blue box) and point wise (red arc); in (c) the results of the abstract and point wise application of $R[\sfrac{\pi}{3}]$ to the abstract and concrete state in (b) respectively. The results in the figures use $\epsilon \approx 0.2$.}\label{fig: interval-2d}
\end{figure}
We observe that applying a rotation to a point generates an arc on the unit circle.  
The interval abstraction, which treats each component independently, encloses this arc within a bounding box that may include points not present in the actual set.  
While the first abstraction in \autoref{fig: interval-approx1}, is complete (the box exactly contains the arc without including any additional points on the circumference), the situation changes after a second rotation (\autoref{fig: interval-approx2}): due to the extra points already introduced by the first abstraction, the new bounding box now includes not only the concrete points in the rotated arc but also additional, spurious points along the unit circle causing incompleteness.

\subsection{On the Completeness of the Abstract Semantics}\label{sec: completeness}
We have seen that completeness cannot hold for all possible inputs, i.e., it fails \emph{globally}~\cite{cousot1979systematic, GRS00}. 
However, we can still investigate whether completeness holds under certain conditions on the operators or by restricting the set of possible inputs.%
\footnote{When completeness is guaranteed only for specific inputs instead of universally, we refer to it as \emph{local completeness}~\cite{bruni2023correctness}.}
This perspective allows us to better understand when and where precision is lost, namely, which inputs and which computations are responsible.

As shown in \autoref{eq: U sem} and \autoref{eq: Up-sem}, the semantics of the statements involve sums of complex numbers that are linked by normalization constraints, which ultimately cause incompleteness. 
However, as we are about to show, by restricting either the type of operations or the structure of the inputs, these sums over complex numbers may simplify to single terms. 
This would eliminate the dependency issues discussed in the previous section and ensure completeness.

\begin{theorem}
    If $\U$ is a {\em generalized permutation}\footnote{A generalized permutation is a matrix with exactly one nonzero entry in each row and each column.}, $\abSem{\U}{}$ is complete.
\end{theorem}
\begin{proof}
    Given a generalized permutation $\U$, for each row of the matrix, there is only one entry different from zero.
    We recall that given a basis $e \in \bases$, the entries  $\{\U_{\tinyd{e,e'}}\}_{\forall e' \in \bases}$ corresponds to one row of the matrix $\U$.
    Thus for each $e$ there is one base $u_e$ that corresponding to only non-zero entry in the row $e$, in other words, $\forall e\in\bases.\;\U_{\tinyd{e,e'}}\neq0$ iff $e'=u_e$.
    Thus, considering the equations in \autoref{eq: U sem}, if $\U$ is a generalized permutation its semantics is $\csSem{\U} = \lambda \psi . \left(\lambda e \in \bases . (\U_{\tinyd{e,u_e}} \cdot \psi(u_e)\right)$, where $u_e$ is the index of the only non-zero entry for each row $e$.
    Similarly, considering the \autoref{def: abssem}, we have $\ok{\abSem{\U}{\abenv} = \lambda \abste . \left(\lambda e \in \bases . (\U_{\tinyd{e,u_e}} \cdot \abste(u_e)\right)}$.
    According to these simplifications, $\abSem{\U}{}$ is complete, i.e., $\forall \Psi \in \wp(\DomP), \Sigma \in \wp(\DomSenv)$ $\alphaP(\csSem{\U}(\Psi)) = \abSem{\code{S}}{\alphaE(\Sigma)}(\alphaP(\Psi))$, if and only if $\alphaP(\sset{\lambda e \in \bases . (\U_{\tinyd{e,u_e})} \cdot \psi(u_e)}{\psi \in \Psi} = 
    \lambda e \in \bases . (\U_{\tinyd{e,u_e}} \cdot (\alphaP(\Psi))(u_e)$ that holds iff $\forall e \in \bases.\; \alphaC\left(\sset{ (\U_{\tinyd{e,u_e}} \cdot \psi(u_e))}{\psi \in \Psi}\right) 
    = \left( \U_{\tinyd{e,u_e}} \cdot \alphaC\left(\sset{\psi(u_e)}{\psi \in \Psi} \right)\right)$.
    Since $\U_{\tinyd{e,u_e}}$ is a number, the equality always holds.
\end{proof}
   
\begin{corollary}
    $\forall k \in \bN$, $\abSem{\CN}{}$, $\abSem{\code{Ry}_{k\pi}}{}$ $\abSem{\code{Rx}_{k\pi}}{}$ $\abSem{\code{Rz}_{k}}{}$, are complete.
\end{corollary}

We now focus on the encoding operators $\Up$, which, unlike the parametric operators $\U$, depend on the environment and thus represent families of operators.
This dependency prevents us from characterizing their completeness by restricting the operator's structure alone.
Instead, we turn our attention to \emph{local completeness}, which can be studied by restricting the class of possible inputs.

\begin{theorem}
    $\abSem{\Up}{}$ is (local) complete on $\Sigma \in \wp(\DomSenv)$ and $\psi \in \DomP$ if $\exists u \in \bases.\; \psi(u) = 1 \wedge \forall e \neq u.\; \psi(e) = 0$, and $\gammaE\circ\alphaE(\Sigma) = \Sigma$ namely, for any $x$ the set $\sset{\sigma(x)}{\sigma \in \Sigma}$ is a closed interval.
\end{theorem}
\begin{proof}
    Let $\abenv = \alphaE(\Sigma)$.
    As before, by simplifying \autoref{eq: Up-sem} and \autoref{def: abssem}, we can say that $\abSem{\code{S}}{\abenv}(\alphaP(\psi))$ is complete if and only if $\forall e \in \bases.$ $\alphaC\left(\sset{\grassto{\Up_{\tinyd{e,u}}}_{\tinyd{\Sigma}} \cdot \psi(u)}{\sigma \in \Sigma} \right) = \agrassto{\Up_{\tinyd{e,u}}}_{\tinyd{\abenv}} \cdot \alphaC(\psi(u))$.
    Since $\psi(u) = 1$, the equation above holds if and only if $\alphaC\left(\sset{ \grassto{\Up_{\tinyd{e,u}}}_{\tinyd{\Sigma}} }{\sigma \in \Sigma} \right) = \agrassto{\Up_{\tinyd{e,u}}}_{\tinyd{\abenv}},$ i.e., $\agrassto{\Up_{\tinyd{e,u}}}_{\tinyd{\abenv}}$ is complete.
    From \autoref{sec: abssem}, we know that $\agrassto{\Up_{\tinyd{e,u}}}_{\tinyd{\abenv}}$ is defined as the $bca$ of $\grassto{\Up_{\tinyd{e,u}}}_{\tinyd{\Sigma}}$. Therefore, if $\gammaE \circ \alphaE(\Sigma) = \Sigma$, then $\agrassto{\Up_{\tinyd{e,u}}}_{\tinyd{\abenv}}$ is trivially complete~\cite{cousot1979systematic,GRS00}.
\end{proof}

Finally, we consider when local completeness holds for the final measurement.

\begin{theorem} 
    If $\forall e \in \bases. \Re(\Psi_{\!\qnn}(e)) = 0 \vee \Im(\Psi_{\!\qnn}(e)) = 0$, $\abSem{\M}{}$
    is (local) complete on $\Psi$.
\end{theorem}
\begin{proof}
    By definition, if for all $\psi \in \Psi$ and for all bases $e$, $\Re(\Psi_{\!\qnn}(e)) = 0 \vee \Im(\Psi_{\!\qnn}(e)) = 0$, the measurement operators simply compute the square of a real numbers.
    The abstraction of the square function is always complete in real interval arithmetic by definition~\cite{intervalbook}.
\end{proof}

\section{(Abstract) VQC-based classifier}\label{sec: vqc classifier}
Similarly to what happens when dealing with NNs, VQCs are primarily employed for data classification tasks~\cite{wendlinger2024comparative}. 
In the previous section, we discussed how to abstract the execution of a VQC when its input is represented as an interval. In this section, we exploit both the concrete and abstract semantics of VQC for modeling the final step of its employment as a classifier: the classification function, i.e., the function deputed to map the output of the quantum computation to a corresponding class label.

\subsection{VQC-based Classification}
A VQC returns a probability distribution over the set of possible results.
Crucially, for classification purposes, the output distribution of a VQC must be post-processed. As illustrated in the example in \autoref{sec: vqc}, classification is determined by the value of the first qubit. 
This means that, for instance, the probabilities of measuring $10$ and $00$ are aggregated into one class, while $11$ and $01$ are aggregated into another. More generally, the number of target classes is often smaller than the number of possible results, requiring a mapping from quantum outputs to class labels by summing the probabilities of grouped outcomes accordingly.

For a $\qnn\in\lang$, we introduce the function $\winner: \DomR \times \bigcup_{n\leq |\varset|}\varset^n \to \bigcup_{n\leq |\varset|}\{0,1\}^n$ that determines which class is the selected one, based on the outcome of the VQC execution.
Let $\rho \in \DomR$ be the resulting distribution of a VQC, and let $\arrw{q} \in \varset^n$ be a tuple of $n$ qubits, we define: 
$$
\winner(\rho, \arrw{q}) \defi \arg\max_{b \in \{0,1\}^{n}} \sum_{\substack{e \in \bases .  e[\arrw{q}] = b}} \rho(e),
$$
where $e[\arrw{q}]$ denotes the projection of the bitstring $e$ onto the positions in $\arrw{q}$.
\\
As an example, consider $\arrw{q} = \tuple{q_0}$, meaning that we are interested in the classification task observing the value of the single ($n=1$) qubit $q_0$. Let us consider the distribution resulting from the VQC in \autoref{sec: vqc}, i.e., $\rho = \{00 \mapsto 0.26,\ 01 \mapsto 0.21,\ 10 \mapsto 0.01,\ 11 \mapsto 0.52\}$, where the first bit corresponds to qubit $q_1$ and the second to $q_0$. In order to determine the selected class, we compute the aggregated probabilities: for $b = 0$, i.e., outcomes where $q_0 = 0$, the relevant bases are $00$ and $10$  ($00[\arrw{q}]=10[\arrw{q}]=0$), with $\rho(00) = 0.26$ and $\rho(10) = 0.01$, hence $\rho(10)+\rho(00) = 0.27$.
Similarly, for $b = 1$, i.e., outcomes where $q_0 = 1$, the relevant bases are $01$ and $11$  ($01[\arrw{q}]=11[\arrw{q}]=1$), with $\rho(01) = 0.21$ and $\rho(11) = 0.52$, so $\rho(01)+\rho(11) = 0.73$.
Since $0.73 > 0.27$, we have: $\winner(\rho, \tuple{q_0}) = 1$.
That is, the most likely class according to the classification rule based on $q_0$ is the one associated with $1$.

We can now define the semantics of the whole VQC classifier.
Without loss of generality, we assume that the execution of a VQC always starts from a specific initial state where all quantum variables are in the value $0$.
Thus, given $\qnn\in\lang$ ($|\varset|=n$), we define the initial state $\init{\psi}\in\DomP$ as:
\begin{equation}\label{eq: init concrete}
\init{\psi}\defi\lambda e\in\bases.\:\left \{
\begin{array}{ll}
     1& e=0^n \\
     0& \mbox{othewise}
\end{array}
\right.
\end{equation}
If qubits $\arrw{q} \in \varset^n$ ($n \leq |\varset|$) are the qubit encoding the classical dataset to be classified,
the semantics of the whole VQC classification algorithm can be modeled as $\winner(\sSem{\qnn}(\init{\psi}),\arrw{q})$. We abuse notation by using the function $\winner$ also to denote its additive lift to sets of distributions:
\begin{equation}\label{eq:concrete_class}
   \winner(\sSem[\Sigma]{\qnn}(\init{\psi}),\arrw{q})\defi\sset{\winner(\sSem{\qnn}(\init{\psi}),\arrw{q})}{\sigma\in\Sigma}. 
\end{equation}

An explicit definition of the classification functions is given in~\cite{schuld2020circuit}, where the authors define a classification function based on the measurement of a single qubit followed by classical post-processing with a tunable bias.
This bias, implemented purely at the classical level by adjusting the decision threshold, allows the model to learn more flexible decision boundaries without increasing the quantum circuit depth. 
For the sake of simplicity but without loss of generality, we omit such bias terms in this paper. Nonetheless, our framework is general enough to accommodate them.

\subsection{Abstract VQC-based classification}

The classification process described above, which assigns a unique class based on the output distribution, applies to concrete executions of the VQC. However, when abstracting the semantics of the VQC, the output is no longer a single probability distribution but an abstract distribution, where each standard basis vector is associated with an interval of possible probabilities. As a consequence, the aggregated probabilities for different classes may result in overlapping intervals, making it impossible to determine a unique selected class with certainty.

A first step toward managing the uncertainty introduced by abstraction is to return a {\em set} of potential selected classes, i.e., all classes that could be chosen given the overlapping probability intervals. This allows us to assess whether all possible classifications are acceptable, or a refinement of the abstraction is necessary to reduce ambiguity, as it has been formalized for verifying abstract robustness in NN~\cite{GMP24,marzari2025advancing}.

Hence, we first abstract the $\winner$ function in order to select only the intervals with a lower bound greater than any other interval distribution. Specifically, $\awinner: \absdist \times \bigcup_{n\leq |\varset|}\varset^n \to \bigcup_{n\leq |\varset|}\wp(\{0,1\}^n)$ selects the set of potential classes with higher probability. 
Formally, given
$\arrw{q}\in\varset^n$
\[
\awinner(\absdist,\arrw{q}) \defi \left \{b \in \{0,1\}^{n} ~\left |~\nexists b'\in \{0,1\}^{n} . \sum_{\substack{e \in \bases.e[\arrw{q}] = b}}\absdist(e)<_{\tinyd{\bRI}} \sum_{\substack{e \in \bases .  e[\arrw{q}] = b'}}\absdist(e)\right.\right\},
\]
where the summation $\sum$ is computed according to the real interval arithmetic. 
Since we are working with abstract distributions, we may have an overlap of classes that can be selected. In fact, we note that the image of the function is in $\wp(\{0,1\}^n)$.  
To see this on an example, let us consider the following abstract distribution: $\absdist = \{000 \to [0.2, 0.46], 001 \to [0, 0.2], 010 \to [0.1, 0.03], 101 \to [0.11, 0.6], 111 \to [0.1, 0.6]\}$ on 3 qubits $\tuple{q_2,q_1,q_0}$ standard basis vectors, where we omit the ones mapped to the singleton $[0,0]$.
Let $\arrw{q} = \tuple{q_2, q_0}$, we first compute $\sum_{\substack{e \in \bases .  e[\arrw{q}] = b}} \absdist(e)$ and we obtain: $[0.29, 0.56]$ for $00$, $[0, 0.2]$ for $01$, $[0.22, 1.2]$ for $11$ and $[0,0]$ for $10$.\footnote{For the sake of readability, we assume that the result represents a map from the qubits of interest to their values, i.e., if $\arrw{q} = \tuple{q_2, q_0}$ then $01$ stands for $q_2\mapsto 0$ and $q_0\mapsto 1$.} In this case, there is no single class whose lower bound exceeds the upper bounds of all other classes. However, we can observe that the configurations $\{01, 10\} \in \wp(\{0,1\}^2)$ have definitively lower probabilities than $\{00, 11\} \in \wp(\{0,1\}^2)$. Therefore, the function $\awinner$ will return the latter set of configurations, whose probability intervals overlap but are still greater than those of the others.

Hence, given $\arrw{q} \in \varset^n$ ($n \leq |\varset|$) consisting in the list of qubits to consider for the classification, and the abstract initial state $\init{\abste} = \alphaP(\{\init{\psi}\})$, the abstract semantics of the classification algorithm is:
\[
\awinner(\abSem{\qnn}{\abenv}(\init{\abste}), \arrw{q})
\]
whose task consists of computing such approximated classification, potentially returning a {\em set} of possible selected results. Then we can observe that abstract classification is also sound.

\begin{proposition}\label{prop: abs ver sound}
   $\forall \Sigma\in\wp(\DomSenv),\arrw{q}\in\varset^n\ (n\leq|\varset|).\:\winner(\sSem[\Sigma]{\qnn}(\init{\psi}),\arrw{q})\subseteq \awinner(\abSem{\qnn}{\alphaE(\Sigma)}(\init{\abste}),\arrw{q})$, i.e., the abstract classification is sound. 
\end{proposition}

\subsection{Robustness Verification of (Abstract) VQC Classifier}
\label{sec: hardness}

In this section, we discuss the robustness verification problem, sometimes referred to as {\em stability}, of VQCs, i.e., the problem of verifying that perturbations of the input do not affect the output classification. The following formalization can be easily extended to the verification of other properties, such as safety, as in classical NN-Verification~\cite{liu2021algorithms}.

Specifically, we focus on input perturbations bounded by an $\ell_\infty$-norm ball centered at a nominal input $x \in \mathbb{R}^n$, formally defined as: $\mathcal{C}_\infty(x, \varepsilon) = \{ x' \in \mathbb{R}^n \mid \|x' - x\|_\infty \leq \epsilon \}.$ In practice, we test the circuit on abstract input intervals in the form $\ok{\Sigma_{\sigma,\epsilon} \defi \sset{\sigma'}{\vert\vert \sigma(x) - \sigma'(x) \vert \vert_\infty \leq \epsilon}}$, where $\sigma$ is the perturbated input (environment) and $\epsilon$ represents the noise applied to such input. 
We now formally define the robustness verification problem for VQC-based classifiers.

\begin{definition}[Robustness Verification of VQC (\textsc{RVVQC})]\label{def:vvqc}
  Let $\qnn\in\lang$, $\epsilon\in\bR$ the input perturbation, and $\arrw{q}\in\parset^n$ ($n\leq|\varset|$), the list of output observed qubits.
  We say that $\qnn$ is \textit{robust}, w.r.t.\ $\arrw{q}$ and $\epsilon$, (denoted $\qnn\models_{\arrw{q},\epsilon}\mbox{\textsc{RVVQC}}$) iff $\forall\sigma \in \DomSenv$
  \begin{equation}\label{eq:rvvqc}
      \winner(\sSem[\Sigma_{\sigma,\epsilon}]{\qnn}(\init{\psi}),\arrw{q}) \subseteq \{\winner(\sSem[\sigma]{\qnn}(\init{\psi}),\arrw{q})\}
  \end{equation}
 \end{definition}
 \noindent
 This definition tells us that a VQC classifier for $\qnn$ is robust whenever the set of resulting classes obtained by executing the VQC classifier on all the considered perturbations of the input $\sigma$ is the singleton set containing only the class obtained by executing the VQC classifier precisely on $\sigma$.

Now we have to understand the computability and complexity concerns of such VQC verification issues.
We can observe that, in the \textsc{RVVQC} problem, the input set (of environments) has the same cardinality as reals, being any possible 
association between (a finite set of) variables and reals
This makes, in general, the verification problem (but also the corresponding confutation problem) not computable. We can observe that this happens independently of the potential recursivity (i.e., decidability in finite time) of the property to verify on the single input. 
This means that, in order to computationally treat the verification problem, we need to restrict the input space to {\em computable real numbers with a fixed maximal precision} only (in practice, this is the set of floating-point numbers), namely those reals 
that can be computed, to within any desired precision, by a finite, terminating algorithm.\footnote{A computable number [is] one for which there is a Turing machine which, given n on its initial tape, terminates with the nth digit of that number [encoded on its tape]~\cite{Minsky67}.} 
This makes the input set space a countable and discrete subset of reals~\cite{Minsky67} and \autoref{eq:rvvqc} decidable, being $\Sigma_{\sigma,\epsilon}$ discrete and finite. In the following part of the paper, we consider as input environments only those defined by mapping variables to computable and discrete reals.

\begin{proposition}
Let $\qnn\in\lang$, $\epsilon\in\bR$, and $\arrw{q}\in\parset^n$ ($n\leq|\varset|$)
\begin{enumerate}
    \item $\gateg{CP}\defi\sset{\qnn}{\qnn \not\models_{\arrw{q},\epsilon} \mbox{\textsc{RVVQC}}}$ is computable (recursive enumerable).
    \item $\gateg{VP}\defi\sset{\qnn}{\qnn \models_{\arrw{q},\epsilon} \mbox{\textsc{RVVQC}}}$ is not computable;
\end{enumerate}
\end{proposition}
The intuition is that, for computing $\gateg{VP}$ we have to execute the VQC on an infinite (even if computable) set of inputs, while for computing $\gateg{CP}$ it is sufficient to find one input not satisfying \autoref{eq:rvvqc}, and if it exists, we can find it in finite time over a discrete space of inputs (e.g., by using a technique called {\em dovetail} for searching on the input space without diverging).
Trivially, if we restrict to a finite set of inputs, then both problems collapse to recursive problems.
Nevertheless, even when made decidable, providing an answer to the \textsc{RVVQC} problem requires testing the VQC on a potentially huge set of inputs (depending on the chosen precision in approximating reals), which is clearly hard. 

\begin{restatable}{theorem}{thsattovqc}\label{thm:3sat-to-vqc}
RVVQC is NP-hard.
\end{restatable}

\begin{proof}

To show the hardness of the problem, we show the following reduction \textsc{3-SAT} $\preccurlyeq_k$ \textsc{RVVQC}. In particular, we will show that any \textsc{3-SAT} instance on a Boolean formula $\Phi$ can be converted into a VQC $\qnn$ such that the property $\prop$ is satisfied on $\qnn$ if and only if $\Phi$ is satisfiable.

Let $\Phi = (C_1 \wedge \ldots \wedge C_m)$ be a 3-conjunctive normal form (CNF) over a set of Boolean variables $X = \{x_1, x_2, \dots, x_n\}$. Each formula is a conjunction of clauses, where each clause is a disjunction of exactly three literals. Each literal is either a variable $x_k$ or its negation $\neg x_k$, for some $k \in \{1, \dots, n\}$. An assignment $a: X \to \{0,1\}$ satisfies $\Phi$ if it satisfies all the clauses simultaneously. 

For the reduction, we first note that $\Phi$ can be viewed as a Boolean circuit which computes a function $f: a \to \{0,1\}$, such that $f(a) = 1$ if and only if $a$ satisfies $\Phi$. Hence, from \cite[Section 7]{Kitaev} we can transform $f$ into a quantum circuit $\gateg{U}_{\mathrm{3-SAT}}$ acting on $\mathrm{poly}(n)$ qubits. 
However, for the complete reduction, since we aim to obtain a VQC from $\Phi$, we first need to introduce an \textit{encoding gadget} $\gateg{E}$ which takes as input $a$ and encodes it for the $\gateg{U}_{\mathrm{3-SAT}}$. In \autoref{fig: encoding sat} we report a schematic overview of $\gateg{E}$.  

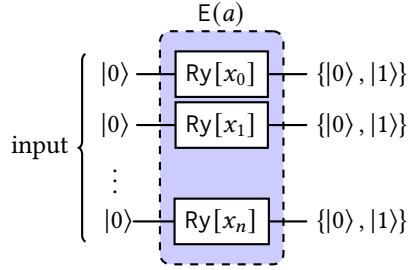
\begin{figure}[h!]
        \centering
        \begin{quantikz}[wire types={q,q,n,q},row sep=0.2em]
            \lstick[4]{input} \ \ket{0} \ & \gate{\gateg{Ry}[x_0]}
            \gategroup[4,style={dashed,rounded
            corners,fill=blue!20, inner
            xsep=2pt},background,label style={label
            position=above,anchor=north,yshift=0.3cm}]
            {{$\gateg{E}(a)$}} 
            & \rstick{\{$\ket{0}, \ket{1}$\}}\\
            \ \ket{0} \ & \gate{\gateg{Ry}[x_1]} & \rstick{\{$\ket{0}, \ket{1}$\}}\\
            \vdots & & \\
            \ \ket{0} & \gate{\gateg{Ry}[x_n]} & \rstick{\{$\ket{0}, \ket{1}$\}} 
        \end{quantikz}
        \caption{\textit{Encoding Gadget} $\gateg{E}$}
        \label{fig: encoding sat}
\end{figure}

In detail, the input of $\gateg{E}$ corresponds to $\vert X \vert$ qubits initialized to $\ket{0}$. Then for a given an assignment $a$ for $\{x_1\to\{0,1\}, \ldots, x_n\to\{0,1\}\}$, we define a static environment $\sigma_a$ as:

\begin{equation}
    \sigma_a(x_i) \gets
    \begin{cases}
    \pi & \text{if } a(x_i) = 1 \\
    0 & \text{if } a(x_i) = 0
    \end{cases}
\end{equation}

For the output of $\gateg{E}$, we then prove the following.

\begin{lemma}\label{lemma: Egate_correctness}
    For any assignment $a:X\to\{0,1\}$, the evaluation of the encoding gadget $\gateg{E}$ w.r.t. the static environment $\sigma_a$ produces an output which is consistent with $a$.
\end{lemma}

\begin{proof}
    The proof is straightforward by observing that if the original assignment for an $x_i$ is $1$, then the gate rotation $\gateg{Ry}[\sigma_a(x_i)=\pi]$ will produce $\ket{1}$ as output. On the other hand, if $x_i=0$ then $\gateg{Ry}[\sigma_a(x_i)=0]$ results in the identity leaving the state as $\ket{0}$.
\end{proof}

By combining the \textit{encoding gadget} $\gateg{E}$ with the $\gateg{U}_{\mathrm{3-SAT}}$ circuit, we obtain a VQC $\qnn$ as desired. Clearly, the reduction from any \textsc{3-SAT} instance $\Phi$ to the VQC $\qnn$ can be carried out in polynomial time. Let $\Phi$ have $n$ variables and $m$ clauses. The encoding gadget $\gateg{E}$ introduces one $R_y$ rotation gate per variable, thus requiring $O(n)$ gates. For each clause, the construction of the $\gateg{U}_{\text{3-SAT}}$ circuit involves a constant number of controlled-NOT and negation gates (to implement 3-literal OR via De Morgan), requiring $O(1)$ gates per clause. The final conjunction over the clause outputs requires $m-1$ additional controlled operations. Hence, the total number of quantum gates and qubits used in $\qnn$ is $O(n+m)$.

To complete the proof, we need to show the following.

\begin{lemma}
    Any \textsc{3-SAT} formula $\Phi$ can be reduced into a VQC $\qnn$ and a property $\prop = \{1\}$, such that\\ $\prop$ is satisfiable on $\qnn$ if and only if $\Phi$ is satisfiable.
\end{lemma}

\begin{proof}
Without loss of generality, we will prove the lemma exploiting the example of the reduction of a simple formula $\Phi$ with $n=3$ variables and $m=2$ clauses into a VQC $\qnn$ reported in \autoref{fig: sat circuit}.

\vspace{2mm}
\boxed{\implies}

We need to show that if $\winner(\sSem[\sigma_a]{\qnn}(\init{\psi}),\arrw{q}) \subseteq \prop$, i.e., the concrete execution of the VQC with the environment $\sigma_a$, then there exists an assignment $a$ such that $\Phi(a) = 1$. Assume that $\winner(\sSem[\sigma_a]{\qnn}(\init{\psi}),\arrw{q}) \subseteq \prop = \{1\}$. This means that when the VQC $\qnn$ is executed and measured, the output qubit is observed in the state $1$. By construction, $\qnn$ outputs $\ket{1}$ if and only if all clause qubits $a_i$ (for $i = 1, \ldots, m$) are set to $\ket{1}$, meaning that each corresponding clause of the formula $\Phi$ is satisfied. Therefore, measuring the output qubit in state $\ket{1}$ implies that the assignment used by the circuit satisfies every clause in $\Phi$. By \autoref{lemma: Egate_correctness}, each assignment $a$ to the Boolean variables of $\Phi$ corresponds to a valid and consistent quantum state prepared by the encoding gadget $\gateg{E}$. Hence, if the output qubit of $\qnn$ is measured as $\ket{1}$, then there must exist at least one assignment $a$ such that $\Phi(a) = 1$, which means that $\Phi$ is satisfiable.

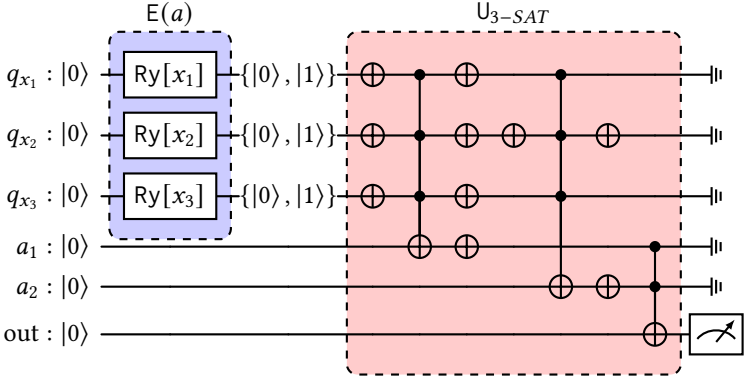
\begin{figure}[t]
        \centering
\begin{quantikz}[row sep=0.2cm, column sep=0.3cm]
\lstick{$q_{x_1}:\ket{0}$} & \gate{\gateg{Ry}[x_1]}\gategroup[3,style={dashed,rounded
            corners,fill=blue!20, inner
            xsep=2pt},background,label style={label
            position=above,anchor=north,yshift=0.3cm}]{$\gateg{E}(a)$} & \push{\{\ket{0},\ket{1}\}}& \targ{}\gategroup[6,steps=7,style={dashed,rounded
            corners,fill=red!20, inner
            xsep=2pt},background,label style={label
            position=above,anchor=north,yshift=0.3cm}]{$\gateg{U}_{3-SAT}$} & \ctrl{3} & \targ{} &          & \ctrl{3} &          & & \ground{} \\
\lstick{$q_{x_2}:\ket{0}$} & \gate{\gateg{Ry}[x_2]} & \push{\{\ket{0},\ket{1}\}}& \targ{} & \ctrl{2} & \targ{} & \targ{} & \ctrl{2} & \targ{} & & \ground{} \\
\lstick{$q_{x_3}:\ket{0}$} & \gate{\gateg{Ry}[x_3]} & \push{\{\ket{0},\ket{1}\}}& \targ{} & \ctrl{1} & \targ{} &          & \ctrl{2} &          & & \ground{} \\
\lstick{$a_1:\ket{0}$}        & & &          & \targ{}  & \targ{} &          &          &          & \ctrl{2}& \ground{} \\
\lstick{$a_2:\ket{0}$}       & & &          &          &          &    & \targ{}  & \targ{} & \ctrl{1}& \ground{}  \\
\lstick{$\text{out}:\ket{0}$} & & &          &          &          &          &          &          & \targ{} & \meter{} 
\end{quantikz}
\caption{Complete reduction of a formula $\Phi(x_1, x_2, x_3) = (x_1 \vee x_2 \vee x_3) \wedge (\neg x_1 \vee x_2 \vee \neg x_3)$ into a VQC $\qnn$.}
\label{fig: sat circuit}
\end{figure}

\vspace{2mm}
\boxed{\impliedby}

Fix an assignment $a$ for $\Phi(x_1, x_2, x_3)$, such that $\Phi(a)=1$. The first step in the reduction is to produce the output of the \textit{encoding gadget} $\gateg{E}$, for the assignment $a$, i.e., $\gateg{E}(a)$. From \autoref{lemma: Egate_correctness} we know that this gadget will produce an output consistent with $a$, for instance if the assignment $a$ produce $\{x_1 \to 0, x_2 \to 1, x_3 \to 1\}$ then the corresponding output of $\gateg{E}(a)$ which becomes the input for $\gateg{U}_{\mathrm{3-SAT}}$ will be $\{\ket{0}, \ket{1}, \ket{1}\}$. 
The $\gateg{U}_{\mathrm{3-SAT}}$ circuit, first applies the negation gates ($\oplus$) to $q_{x_1}, q_{x_2}, q_{x_3}$ to get their negations. By applying a multi-controlled NOT, it computes $(\neg q_{x_1} \wedge \neg q_{x_2} \wedge \neg q_{x_3})$, storing the result in $a_1$. By negating the result of $a_1 = \neg(\neg q_{x_1} \wedge \neg q_{x_2} \wedge \neg q_{x_3}) = (q_{x_1} \vee q_{x_2} \vee q_{x_3})$ it encodes the OR operation for the first clause, as performed in $\Phi$. Similarly, it is done for the second clause, so the circuit first restores the original values of $q_{x_1}, q_{x_2}, q_{x_3}$ by performing a negation on all the values. It then compute the negation only on $q_{x_2}$, storing the result in $a_2 = (x_1 \wedge \neg x_2 \wedge x_3)$. After this, it exploits again a multi-controlled NOT to get the result of $a_2 = \neg(x_1 \wedge \neg x_2 \wedge  x_3) = (\neg x_1 \vee x_2 \vee \neg x_3)$. Finally, to encode the AND of the clauses, it simply performs the $\oplus$ operation on $a_1$ and $a_2$, storing the final result in $out = (x_1 \vee x_2 \vee x_3) \wedge (\neg x_1 \vee x_2 \vee \neg x_3)$, which is precisely the original formula $\Phi$. Hence if $a$ satisfies $\Phi$ then it also necessarily holds by construction that $a_1, a_2 = \ket{1}$, which once measured, will produce $1 \subseteq \prop$.

\end{proof}

The proof is complete.

\end{proof}

To address this issue, we investigate whether our abstract interpretation framework can, in some sense, mitigate the inherent complexity of the problem. 
From \autoref{prop: abs ver sound}, we know that classification based on the abstract semantics of a VQC is sound. This soundness ensures conservativeness with respect to the concrete computations—that is, no concrete behaviors are missed, and thus no possible classifications are lost. 
Hence, we can use the computed over-approximation in order to verify a desirable property \textsc{RVVQC}. Specifically, if the abstract output assigns a strictly higher probability to one class compared to all others, then—by the soundness of the abstraction—that class is guaranteed to be selected in the concrete execution as well.

\begin{proposition}[Abstract \textsc{RVVQC}]
    \label{def:verification}
    Consider a VQC $\qnn$ and $\arrw{q}\in\parset^n$ ($n\leq|\varset|$).
    Given the initial abstract state $\init{\abste}$.
    Then $\qnn$ is \textit{robust} if $\forall\sigma\in\DomSenv$ $$\awinner(\abSem{\qnn}{\alphaE(\Sigma_{\sigma,\epsilon})}(\init{\abste}),\arrw{q}) \subseteq \{\winner(\sSem[\sigma]{\qnn}(\init{\psi}),\arrw{q})\}.$$
\end{proposition}

However, due to the over-approximation introduced by abstraction, the abstract semantics may produce overlapping probability intervals across classes, making the classification result potentially uncertain and the verification outcome potentially inconclusive. 
%

As a result, the potential overlap in the abstract output can distort the classification outcome, falsely declaring a VQC as potentially unsafe (i.e., the function cannot return a single class), even when it is actually safe, a phenomenon commonly referred to as a false alarm due to abstraction.
Moreover, in general, from the computational complexity point of view, we cannot guarantee any improvement, as stated by the following result.

\begin{proposition}
    Abstract \textsc{RVVQC} is NP-hard.
\end{proposition}

Nonetheless, as in classical NN-Verification, abstract interpretation plays a crucial role in providing {\em provable} guarantees on the classifier. Indeed, we may reduce the number of executions to perform since each abstract execution captures several concrete executions, but by formally tackling the issue, we have an even more important gain. In fact, by using abstract interpretation, we can exploit the framework to {\em formally reason} about precision, the main source of uncertainty. In this context, precision is a fundamental challenge when dealing with approximation, and if we are able to identify the sources of imprecision, we may be able to potentially reduce it, reducing, as a consequence, uncertainty. 

\section{Recovering Precision in Abstract VQC Verification}\label{sec: recovery}
As discussed in the previous section, the interval-based abstract semantics is sound but incomplete. 
In this section, we introduce some techniques to mitigate the resulting loss of precision.

\begin{figure}
    \centering
    \begin{subfigure}[c]{0.48\textwidth}
    \centering
    \includegraphics[width=.5\textwidth]{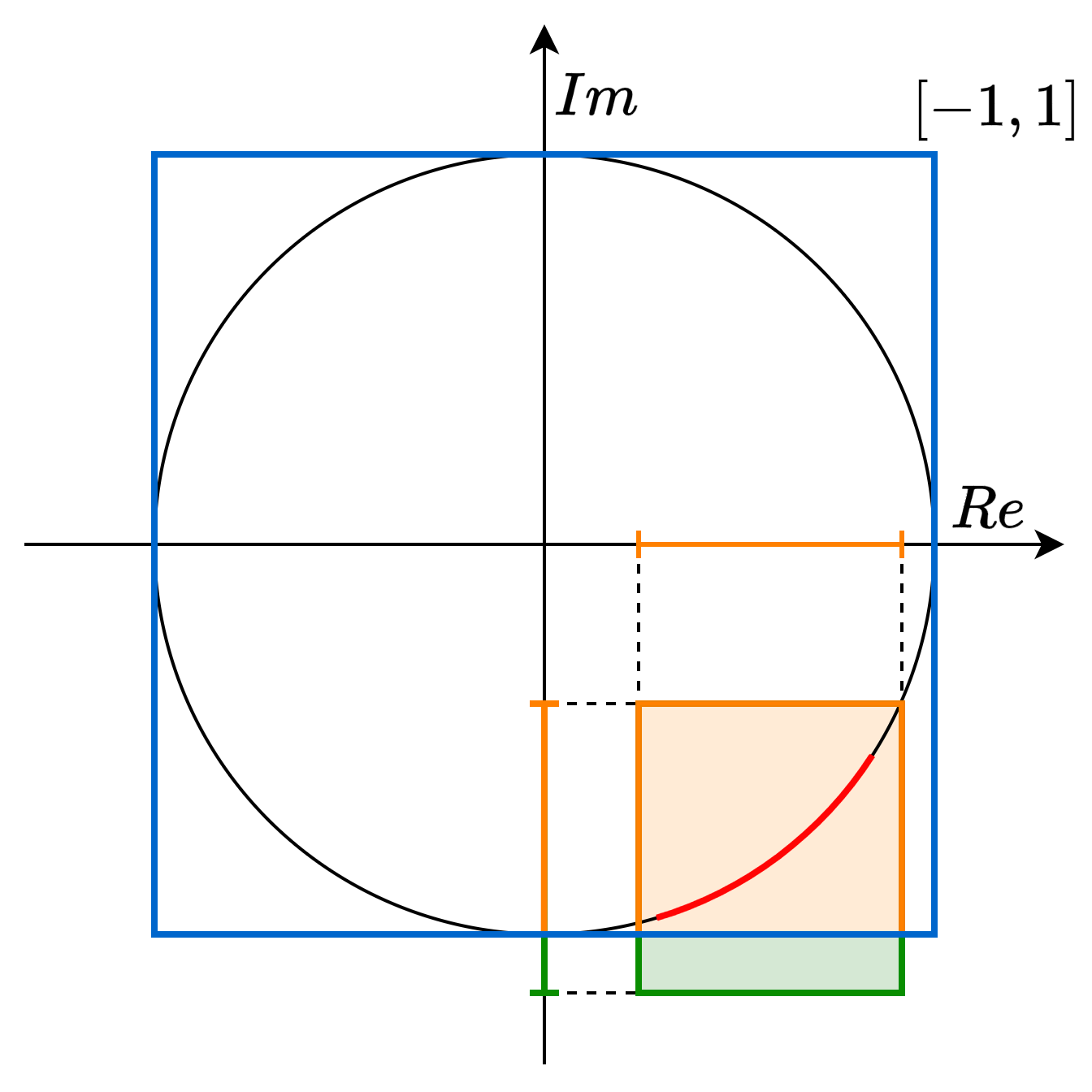}
    \caption{}\label{fig: interval_clipping}
    \end{subfigure}
    \begin{subfigure}[c]{0.48\textwidth}
    \centering
    \includegraphics[trim=20cm 20cm 0em 0em, clip,width=.5\textwidth]{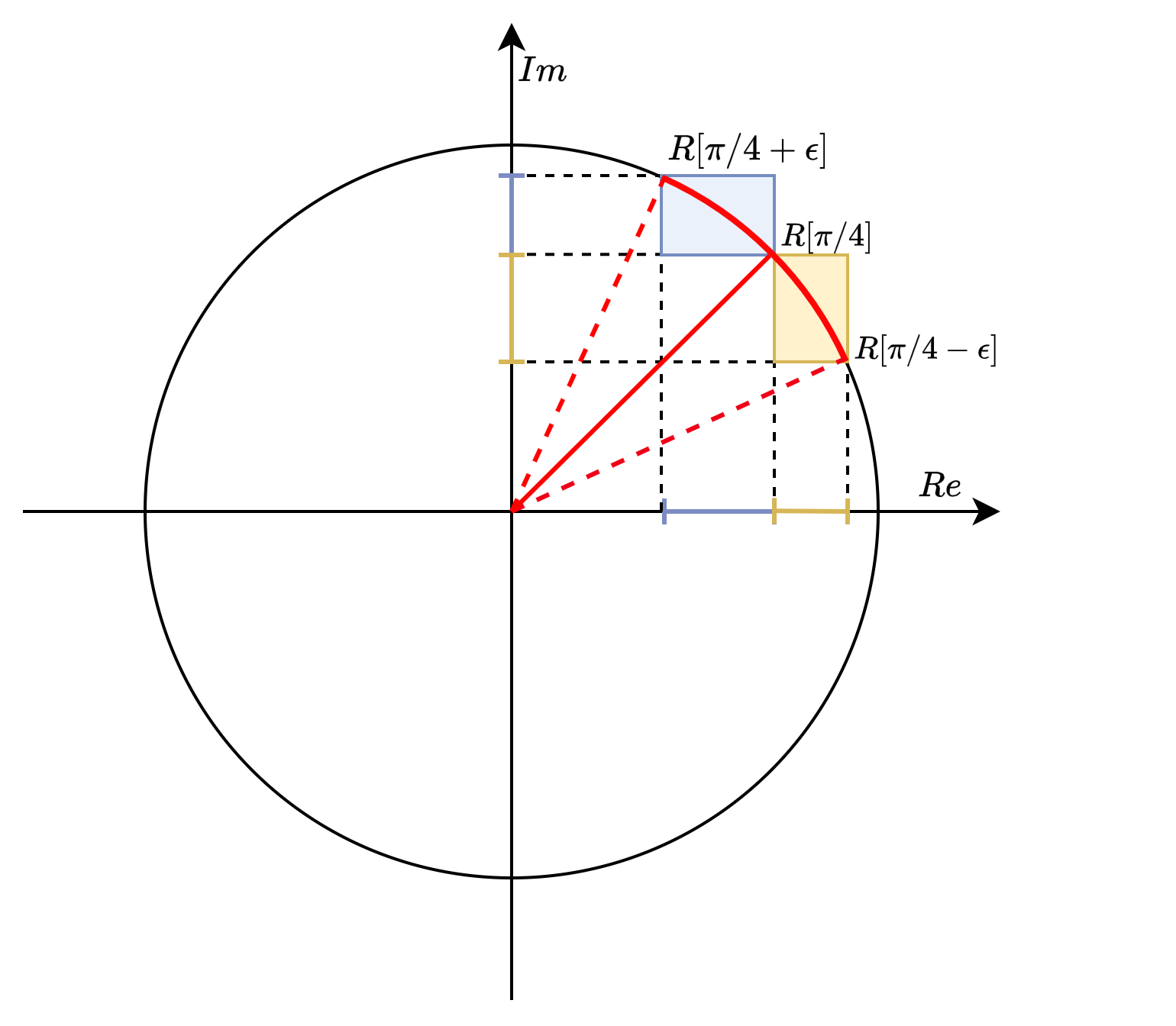}
    \caption{}\label{fig: interval_refinement}
    \end{subfigure}
    \caption{An example of the clipping function (a), and of the refinement based on splitting intervals (b)}
\end{figure}

\paragraph{Clipping the Intervals}
In \autoref{ex: incompletness}, we observe that the abstraction introduces components with norms greater than one. The problem is in the imaginary part of the amplitude associated with the basis $1$ exceeding $1$, due to the over-approximation error.
To mitigate this issue, a possible solution is to remove such invalid elements by clipping all intervals to lie within $[-1, 1]$. That is, at each step of the computation, we intersect each component of the abstract state with the interval $[-1, 1]$.
Formally, we define a function \(\clipC: \aDomP \to \aDomP\) as \(\ok{\clipC \defi \lambda \abste. \left(\lambda e \in \bases. \tuple{\Re(\abste(e)) \cap [-1, 1], \Im(\abste(e)) \cap [-1, 1]} \right) }\).
A 2D illustration is shown in \autoref{fig: interval_clipping}. 
\vspace{-.4cm}
\begin{proposition}\label{prop:clipsound}
    The clip operation is sound for all quantum statements, i.e., $\forall\Sigma\in\wp(\DomSenv)$, $\alphaC (\csSem{\nonter{s}}(\init{\psi})) \dot{\leq}_{\tinyd{\bCI}} \clipC(\abSem{\nonter{s}}{\alphaC(\Sigma)}(\init{\abste}))$ 
\end{proposition}
\begin{proof}
    The soundness follows from that given a $\psi \in \csSem{\nonter{s}}(\init{\psi})$, 
    $\forall e \in \bases$, $-1 \le \Re(\psi(e)) \ge 1$ and $-1 \le \Im(\psi(e)) \ge 1$.
    Thus the elements eliminated by the $\clipC$ are not in $\alphaC (\csSem{\nonter{s}}(\init{\psi}))$.
\end{proof}

Similarly, after performing a measurement, the resulting probability intervals may have upper bounds exceeding 1. Since measurement outcomes must lie within the range 
$[0,1]$, we can define another function that applies a clipping by intersecting the interval with 
$[0,1]$. \\ Formally, we can define a function \(\clipR: \aDomR \to \aDomR\) as  \(\ok{\clipR \defi \lambda \absdist. \left(\lambda e \in \bases. \absdist(e) \cap [0, 1] \right) }\). For similar consideration as in \autoref{prop:clipsound}, we have

\begin{proposition}
    The clip operation on the measurement is sound. 
\end{proposition}

Since both clip functions are sound, we can guarantee that applying the clip operation at each step of the abstract semantics is sound.

\paragraph{`Symbolic' Execution of Parametric Circuit}
Another technique employed in classical NN-Verification for reducing the over-approximation is called \textit{symbolic interval propagation}~\cite{reluval}. 
At a high level, the idea is to track dependencies between variables symbolically, during abstract computation, reducing over-approximation compared to naive interval propagation methods. 
However, its precision drops in the presence of non-convex transformations, and it often requires interval concretization to remain tractable.

Inspired by this approach, we investigate whether this simple yet effective technique can also be applied in a quantum setting. 
We begin by observing that, due to the linearity of the parametric part of the VQC, the evolution of a concrete quantum state follows the standard rule of operator composition, i.e., matrix multiplication. 
Specifically, given a vector state $\ket{\psi}$ and two unitary operators $\code{U}$ and $\code{U}'$, $\code{U}' \cdot (\code{U} \cdot \ket{\psi}) = (\code{U}' \cdot \code{U}) \cdot \ket{\psi}$. 
On the other hand, when working with abstract states, given an abstract state $\abste$ and two operations $\code{U}$ and $\code{U}'$, the abstract semantics satisfy:
$\abSem{\code{U}' \cdot U}{}(\abste) ,\dot{\leq}_{\tinyd{\bCI}} \abSem{\code{U}'}{} \circ \abSem{U}{} (\abste)$. 
This inclusion reflects the fact that by composing $\code{U}$ and $\code{U}'$ first, we can perform the concrete operations exactly (without over-approximation), and only afterwards apply a single abstract operation to the result, thus reducing the accumulation of over-approximation errors in the analysis.

\begin{example}\label{ex: symb}
In order to understand the idea in practice, consider the example in \autoref{sec: vqc ab example}.  
In this case, if we first compute the full operator $\ok{O \defi (\code{Ry}^{\tinyd{q_1}}_{-0.69} \cdot \code{Ry}^{\tinyd{q_0}}_{3.27}) \cdot \gateg{CX}^{\tinyd{q_0,q_1}} \cdot \gateg{Ry}^{\tinyd{q_1}}_{0.99} \cdot \gateg{Ry}^{\tinyd{q_0}}_{-0.50}}$ and then apply it to the initial abstract state, we can directly compute the final distribution as:
\begin{equation}\label{eq: symb}
    \absdist_s = \abSem{\M}{}\circ\abSem{O}{\abenv}(\abste_1) = \begin{Bmatrix} 
00 \mapsto [0.165, 0.394] & 
01 \mapsto [0.128, 0.410]\\ 
10 \mapsto [0, 0.068] & 
11 \mapsto [0.320, 0.698]
\end{Bmatrix}.
\end{equation}


Thus, we find $q_0$ equals to $0$ w.p. $[0.165,0.462]$ and equals to $1$ w.p. $[0.448, 1.108]$.
We clearly observe that the intervals are narrower compared to those in $\absdist$ from \autoref{ex:absdist}, which was obtained by applying the operators step by step.
However, we are currently unable to verify that the VQC is robust with respect to the initial input $x = [6, 2.7]$ perturbed by an $\epsilon = 0.5$.
We can improve the abstraction process by introducing another technique, called iterative refinement.
\end{example}

\paragraph{Iterative Refiment}

In classical NNs, one common technique to improve the precision of interval analysis is to recursively split the input interval into smaller sub-intervals~\cite{reluval}. For instance, analyzing the network over a broad input range such as $[0, 10]$ might produce a loose output bound like $[2, 30]$. However, if we divide the input into two smaller intervals, say $[0, 5]$ and $[5, 10]$, and analyze each one separately, we might obtain tighter bounds such as $[2, 12]$ and $[10, 20]$, respectively.
Taking the union of these results yields a refined output bound of $[2, 20]$, improving the overall precision of the analysis.
This strategy helps reduce overestimation caused by input dependencies. Its effectiveness is supported by the fact that NNs are Lipschitz continuous, which ensures that repeated input refinement will eventually converge to arbitrarily tight output bounds within a finite number of steps.

In our setting, the encoding functions used to generate quantum states, such as sine and cosine, are continuous. This continuity allows us to adopt a similar strategy; by partitioning the input domain and analyzing smaller subregions independently, we can derive more precise output intervals.
In \autoref{fig: interval_refinement}, we provide a 2D visualization of this idea. 
Instead of computing the abstract rotation $\gateg{R}[\sfrac{\pi}{4}-\epsilon,\sfrac{\pi}{4}+\epsilon]$ all at once, we split the input interval and compute separately $\gateg{R}[\sfrac{\pi}{4}-\epsilon,\sfrac{\pi}{4}]$ and $\gateg{R}[\sfrac{\pi}{4},\sfrac{\pi}{4}+\epsilon]$. 
As shown in the figure, the union of these two results yields a tighter over-approximation than applying the rotation over the entire interval at once.

Once again, to provide a practical example of this technique, let us consider again the VQC $\qnn$ used in \autoref{sec: vqc ab example}. 
To further improve precision, we can refine the input domain by splitting one of the two input intervals, say, $x_0$.
We thus define two sub-environments: $\abenv_0 \colon \{ x_0 \mapsto [5.5, 6],\ x_1 \mapsto [2.2, 3.2] \}$, $\abenv_1 \colon \{ x_0 \mapsto [6, 6.5],\ x_1 \mapsto [2.2, 3.2] \}$, such that $\abenv = \abenv_0 \cup \abenv_1$.
Executing the abstract semantics on these two environments, using symbolic execution, we obtain:
\[
\begin{aligned}
\small
\abSem{\qnn}{\abenv_0} = \begin{Bmatrix}
    00 \mapsto [0.176, 0.387] & \!\!\!\!01 \mapsto [0.139, 0.355] \\
    10 \mapsto [0, 0.057] & \!\!\!\!11 \mapsto [0.322, 0.674]
\end{Bmatrix},\
\abSem{\qnn}{\abenv_1} = \begin{Bmatrix}
    00 \mapsto [0.195, 0.285] & \!\!\!\!01 \mapsto [0.148, 0.367] \\
    10 \mapsto [0, 0.048] & \!\!\!\!11 \mapsto [0.40, 0.635]
\end{Bmatrix}.
\end{aligned}
\]
Taking the union, we obtain:
\[
\abSem{\qnn}{\abenv_0} \cup \abSem{\qnn}{\abenv_1} =
\begin{Bmatrix}
    00 \mapsto [0.176, 0.387] & 01 \mapsto [0.139, 0.367] \\
    10 \mapsto [0, 0.057] & 11 \mapsto [0.322, 0.674]
\end{Bmatrix}.
\]
This yields a more precise output compared to \autoref{eq: symb}.
If we check the results of the VQC, the probability associated with measuring $q_0$ equals $0$ is $[0.176, 0.444]$ while we have $q_0$ equal to $1$ with a probability in $[0.461, 1.04]$, $1$. 
Refining the abstraction confirms that the VQC in \autoref{fig: VQC_example} is robust when the initial input $x = [6, 2.7]$ is perturbed by an $\epsilon = 0.5$.

\paragraph{Why not a `Symbolic Encoding'?}
This raises the question of whether performing a symbolic execution of the encoding part of the circuit could lead to improved precision. As we will show below, unfortunately, the answer is no.

\begin{proposition}
    The symbolic execution of the encoding part of a VQC does not improve precision.
\end{proposition}

\begin{proof}
To prove the proposition, we show a counterexample.
Let us consider the program $s := \code{Rx}^q[x_0];\code{Ry}^q[x_1]$.
We recall that the $\code{Ry}$ is defined as: $\small\code{Ry} = \begin{pmatrix} \cos(\sfrac{\theta}{2}) & - \sin(\sfrac{\theta}{2}) \\
\sin(\sfrac{\theta}{2}) & \cos(\sfrac{\theta}{2})
\end{pmatrix}$ and the $\code{Rx}$ definition is given in \autoref{eq: Rx}.
The small program $s$ models a circuit in which we are encoding multiple classical data in the same qubit.
Let us consider $x = [\pi/2 - \epsilon, \pi/2 + \epsilon]$ and $y = [\sfrac{\pi}{3} - \epsilon, \sfrac{\pi}{3} + \epsilon]$, where $\epsilon = 1$.
First, we can execute the circuit naively, starting from the abstract state $\abste_0 = \{0 \mapsto \tuple{[1,1][0,0]}, 1 \mapsto \tuple{[0,0],[0,0]}\}$.
In particular,
$\abste_1 = \abSem{\code{Ry[y]}^q}{\abenv}\abSem{\code{Rx}^q[x]}{\abenv}(\abste_0) = \begin{Bmatrix} 0 \mapsto \tuple{[0.564, 0.66] + [0.306, 0.402]}, 1 \mapsto \tuple{[0.306, 0.402] + [-0.66, -0.564]} \end{Bmatrix}.$

We now execute the circuit keeping $x$ and $y$ symbolic.
First we compute $\abSem{\code{Rx}^q[x]}{\abenv}(\abste_0)$, that results in the symbolic state: 
$
\abste[T]_1[x] = \begin{Bmatrix} 0 \mapsto \cos(\sfrac{x}{2}) \\ 1 \mapsto -\ii\cdot\sin(\sfrac{x}{2}) \end{Bmatrix}
$.
Then we compute symbolically 
    $\abSem{\code{Ry[y]}^q}{\abenv}(\abste[T]_1[x]) = \small\begin{bmatrix}
\cos(\sfrac{y}{2})\cdot\cos(\sfrac{x}{2})+(-\sin(\sfrac{y}{2}))\cdot(-\ii\sin(\sfrac{x}{2})) \\
\sin(\sfrac{y}{2})\cdot\cos(\sfrac{x}{2})+\cos(\sfrac{y}{2})\cdot(-\ii\sin(\sfrac{x}{2}))
\end{bmatrix}$.
By substituting the values of $x$ and $y$ into the resulting state, we obtain $\abste_1$.  
Alternatively, we can use trigonometric formulas to manipulate the state, obtaining:
$
\abste[T]_2[x,y] = \begin{Bmatrix}
    0 \mapsto \tuple{[\sfrac{\cos\left(\frac{x + y}{2}\right) + \cos\left(\frac{x - y}{2}\right)}{2}], [\sfrac{\cos\left(\frac{x - y}{2}\right) - \cos\left(\frac{x + y}{2}\right)}{2}]}  \\
    1 \mapsto \tuple{[\sfrac{\sin\left(\frac{x + y}{2}\right) - \sin\left(\frac{x - y}{2}\right)}{2}], [- \sfrac{\sin\left(\frac{x - y}{2}\right) - \sin\left(\frac{x + y}{2}\right)}{2}]}
\end{Bmatrix}.
$
Now substituting the $x$ and the $y$ with the proper intervals, we have:
$
\!\abste[T] \!=\! \big\{
    0\! \mapsto\! \tuple{[0.548,\! 0.670] [0.291,\!0.413]},$ $
    1 \mapsto \tuple{[0.290, 0.413] [-0.670, -0.548]}\big\}.$ that is noisier than $\abste_1$.
\end{proof}

We have shown, with a simple example, that the symbolic encoding does not always improve accuracy. 
This is due to the fact that trigonometric substitutions do not always improve precision. 

\section{Evaluation}\label{sec: evaluation}
In this section, we evaluate our abstract interpretation-based framework for the robustness verification of VQC classifiers.
In particular, we start by considering two standard dataset benchmarks in the classification literature, namely Iris~\cite{iris} and MNIST~\cite{MNIST}.
To keep the paper self-contained, we briefly summarize the main characteristics of these datasets here. The Iris dataset consists of 150 samples of iris flowers classified into three species, such as \textit{Iris setosa}, \textit{Iris virginica}, and \textit{Iris versicolor}.
The input to the classifier consists of four features, such as the length and width of the sepals and petals, all in centimeters, respectively.
In our setting, we only consider a binary classification task between \textit{Iris setosa} and \textit{Iris versicolor}.
The MNIST dataset of handwritten digits, containing 60,000 training and 10,000 testing examples, each represented as a grayscale image of size $28 \times 28$ pixels.
In our setting, we consider the binary classification of downscaled to $4\times4$ image digits 0 and 1 (denoted \texttt{MNIST[0,1]}), and digits 2 and 6 (denoted \texttt{MNIST[2,6]}).

Although verification typically assumes a pre-trained model, in order to assess how different levels of accuracy affect the robustness of the VQC, we train and verify three distinct models using our proposed framework as VQC-based classifiers. Specifically: 
\begin{itemize}
    \item \textbf{QCL}~\cite{qcl} is a representative hybrid classical-quantum framework that employs a VQC composed of a nonlinear quantum encoding circuit for input data and a low-depth variational quantum circuit. 
    We train and verify this model on Iris.
    The circuit is represented in \autoref{fig: QCL_model}.
    
    \item \textbf{CCQC}~\cite{schuld2020circuit} is a low-depth hybrid VQC framework designed for supervised learning. See \autoref{fig: CCQC_model} for the circuit representation.  
    Unlike the previous approach, classification is performed based on the measurement outcomes combined with a trainable bias term.  
    We train and verify this model on the Iris dataset. The circuit is represented in \autoref{fig: CCQC_model}.
    
    \item The final models~\cite{huang2024postvar}, which we refer to as \textbf{PV}, are a variation of the previous approaches and employ angle encoding~\cite{munikote2024comparingencoding} and a parametric circuit to mitigate flat cost function landscapes during optimization~\cite{Grant2019initialization}. 
    We train this model on MNIST.
    The circuit is represented in \autoref{fig: PV_model}.
\end{itemize}
 
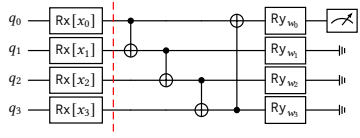
\begin{figure}
    \centering
    \resizebox{0.35\linewidth}{!}{
    \begin{quantikz}[row sep=2pt, slice style=black]
\lstick{$q_0$} & \gate{\gateg{Rx}[x_0]}\slice{} & \ctrl{1} &  &  & \targ{} & \gate{\gateg{Ry}_{w_0}} & \meter{}\\
\lstick{$q_1$} & \gate{\gateg{Rx}[x_1]} & \targ{} & \ctrl{1} &  &  & \gate{\gateg{Ry}_{w_1}} & \ground{}\\
\lstick{$q_2$} & \gate{\gateg{Rx}[x_2]} &  & \targ{} & \ctrl{1} &  & \gate{\gateg{Ry}_{w_2}} & \ground{}\\
\lstick{$q_3$} & \gate{\gateg{Rx}[x_3]} &  &  & \targ{} & \ctrl{-3} & \gate{\gateg{Ry}_{w_3}} & \ground{}
    \end{quantikz}
    }
    \caption{The QCL model, used to classify a 4-feature input data.}
    \label{fig: QCL_model}
\end{figure}

\begin{figure}
    \centering
    \resizebox{.9\linewidth}{!}{\begin{quantikz}[row sep=1.5pt, slice style=black]
\lstick{$q_0$} & \gate{\gateg{Ry}[x_0]} & \ctrl{1} &     & \ctrl{1} & \gate{\gateg{X}} & \ctrl{1} &     & \ctrl{1} & \gate{\gateg{X}}\slice{} & \gate{\gateg{R3}_{w_{0},w_{1},w_{2}}} & \ctrl{1} & \gate{\gateg{R3}_{w_{6},w_{7},w_{8}}} & \ctrl{1} & \gate{\gateg{R3}_{w_{12},w_{13},w_{14}}} & \ctrl{1} & \gate{\gateg{R3}_{w_{18},w_{19},w_{20}}} & \ctrl{1} & \gate{\gateg{R3}_{w_{24},w_{25},w_{26}}} & \ctrl{1} & \gate{\gateg{R3}_{w_{30},w_{31},w_{32}}} & \ctrl{1} & \meter{} \\
\lstick{$q_1$} &     & \targ{} & \gate{\gateg{Ry}[x_1]} & \targ{} & \gate{\gateg{Ry}[x_2]} & \targ{} & \gate{\gateg{Ry}[x_3]} & \targ{} & \gate{\gateg{Ry}[x_4]} & \gate{\gateg{R3}_{w_{3},w_{4},w_{5}}} & \targ{} & \gate{\gateg{R3}_{w_{9},w_{10},w_{11}}} & \targ{} & \gate{\gateg{R3}_{w_{15},w_{16},w_{17}}} & \targ{} & \gate{\gateg{R3}_{w_{21},w_{22},w_{23}}} & \targ{} & \gate{\gateg{R3}_{w_{27},w_{28},w_{29}}} & \targ{} & \gate{\gateg{R3}_{w_{33},w_{34},w_{35}}} & \targ{} & \ground{}
    \end{quantikz}}
    \caption{The CCQC model used on the iris dataset. Here $\gateg{R3}_{w_{i},w_{j},w_{k}} = \gateg{Rz}_{w_{k}}\cdot\gateg{Ry}_{w_{j}}\cdot\gateg{Rx}_{w_{i}}$}
    \label{fig: CCQC_model}
\end{figure}
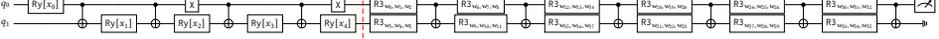

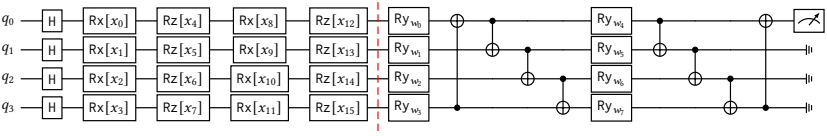
\begin{figure}
    \centering
    \resizebox{.8\linewidth}{!}{
\begin{quantikz}[row sep=1.5pt, slice style=black]
\lstick{$q_0$} & \gate{\gateg{H}} & \gate{\gateg{Rx}[x_{0}]} & \gate{\gateg{Rz}[x_{4}]} & \gate{\gateg{Rx}[x_{8}]} & \gate{\gateg{Rz}[x_{12}]}\slice{} & \gate{\gateg{Ry}_{w_{0}}} & \targ{} & \ctrl{1} &  &  & \gate{\gateg{Ry}_{w_{4}}} & \ctrl{1} &  &  & \targ{} & \meter{} \\
\lstick{$q_1$} & \gate{\gateg{H}} & \gate{\gateg{Rx}[x_{1}]} & \gate{\gateg{Rz}[x_{5}]} & \gate{\gateg{Rx}[x_{9}]} & \gate{\gateg{Rz}[x_{13}]} & \gate{\gateg{Ry}_{w_{1}}} &  & \targ{} & \ctrl{1} &  & \gate{\gateg{Ry}_{w_{5}}} & \targ{} & \ctrl{1} &  &  & \ground{} \\
\lstick{$q_2$} & \gate{\gateg{H}} & \gate{\gateg{Rx}[x_{2}]} & \gate{\gateg{Rz}[x_{6}]} & \gate{\gateg{Rx}[x_{10}]} & \gate{\gateg{Rz}[x_{14}]} & \gate{\gateg{Ry}_{w_{2}}} &  &  & \targ{} & \ctrl{1} & \gate{\gateg{Ry}_{w_{6}}} &  & \targ{} & \ctrl{1} &  & \ground{} \\
\lstick{$q_3$} & \gate{\gateg{H}} & \gate{\gateg{Rx}[x_{3}]} & \gate{\gateg{Rz}[x_{7}]} & \gate{\gateg{Rx}[x_{11}]} & \gate{\gateg{Rz}[x_{15}]} & \gate{\gateg{Ry}_{w_{3}}} & \ctrl{-3} &  &  & \targ{} & \gate{\gateg{Ry}_{w_{7}}} &  &  & \targ{} & \ctrl{-3} & \ground{} 
    \end{quantikz}}
    \caption{The PV model, used to classify a 16-feature input data.}
    \label{fig: PV_model}
\end{figure}

\paragraph{Training Phase.} As described in Sec. \ref{sec: vqc}, a VQC, like a deep NN, consists of a parameterized quantum circuit whose parameters are optimized during training to minimize a cost function, typically defined as the margin between the predicted and target classes~\cite{biamonte2017quantum}. 
For clarity, the left part of Fig. \ref{fig:training_results} presents a simplified version of the training algorithm used for the QCL model on the Iris dataset, while on the right side of the figure, we provide an illustrative view of the decision boundaries for the CCQC model on the same dataset. 
Specifically, we adopt a supervised training procedure for each VQC tested, exploiting the Pennylane framework~\cite{bergholm2018pennylane}. 
For each training example, the prediction is computed via a VQC, and the parameters $\theta$ are updated using gradient-based optimization. 
This process is repeated over multiple epochs to improve classification performance. 
While the pseudocode illustrates a generic Stochastic Gradient Descent (SGD) update, in our experiments, we employ different optimization strategies, such as PennyLane’s \textit{NesterovMomentumOptimizer}~\cite{bergholm2018pennylane}, using the quantum state encoding defined in~\cite{mottonen2004transformation}. The right table in Fig. \ref{fig:verification_results} shows the test set accuracy for each model tested on its corresponding dataset.

\begin{figure}[h!]
  \centering

  \begin{minipage}[t]{0.47\textwidth}
    \centering
    \vspace{-5cm}
    \scriptsize
    \begin{algorithm}[H]
    \caption*{SGD for QCL Binary Classification}
    \begin{algorithmic}[1]
    \State \textbf{Input:} QCL with parameter $\theta$, Training data $\{(X_i, Y_i)\}_{i=1}^N$, learning rate $\eta$, epochs $E$
    \State \textbf{Initialize:} Random parameters $\theta$
    \For{epoch $= 1$ to $E$}
        \For{each $(X_i, Y_i)$ in training set}
            \State $\hat{y}_i \gets \textsc{QCL}(X_i, \theta)$
            \State $loss \gets (\hat{y}_i - Y_i)^2$
            \State $g \gets \nabla_\theta \ell_i$ \Comment{Compute gradient}
            \State $\theta \gets \theta - \eta \cdot g$ \Comment{Update parameters}
        \EndFor
        \State $a \gets \textsc{Accuracy}(\text{training set}, QCL, \theta)$
    \EndFor
    \State \Return trained parameters $\theta$
    \end{algorithmic}
    \end{algorithm}
  \end{minipage}%
  \hfill
  \begin{minipage}[t]{0.47\textwidth}
    \centering
    \includegraphics[width=\linewidth]{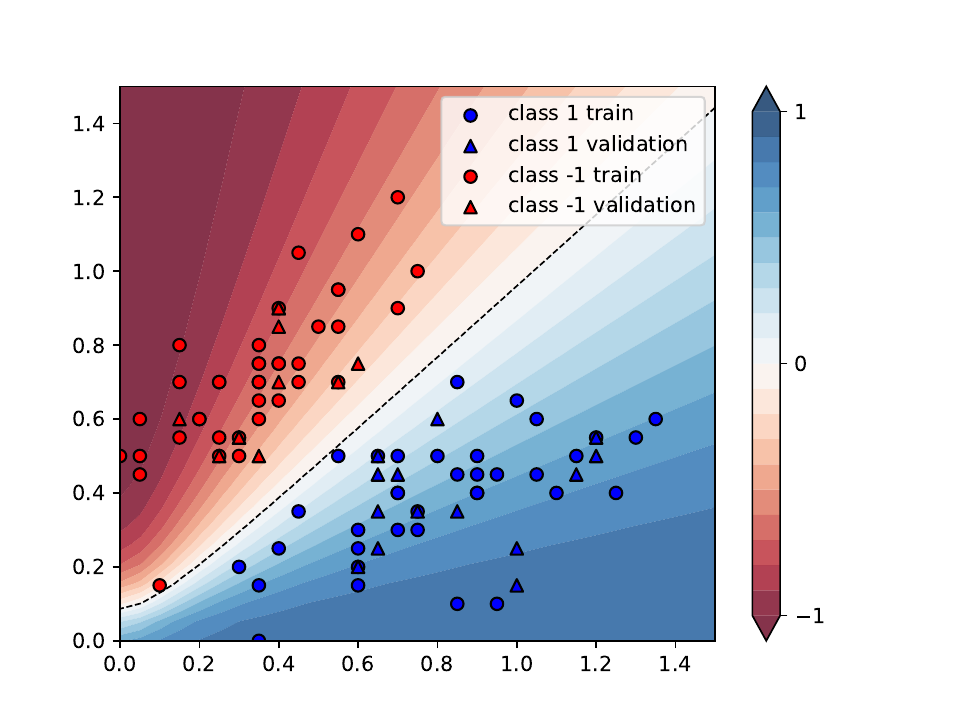}
  \end{minipage}

  \caption{Left: Training algorithm for a Variational Quantum Classifier using SGD. Right: Empirical decision boundary on the Iris dataset using CCQC classifier (100\% accuracy).}
  \label{fig:training_results}
\end{figure}

\paragraph{Verification Phase} 
After the training phase, we perform a formal verification step to assess the robustness of trained VQCs to input perturbations. 
To this end, we extend the Pennylane framework~\cite{bergholm2018pennylane} to support our abstract interpretation-based formulation, which operates over intervals of real and complex numbers. 
This extension enables reachability analysis through the circuit, allowing us to reason about how input uncertainty propagates through quantum operations.

The verification process is summarized on the left side in \autoref{fig:verification_results}. 
In detail, we show the algorithm used to compute the maximum perturbation $\epsilon$ each model can tolerate for a given input $x$ and target class $y$. 
The \textsc{Verify} function is the core of the pipeline. It takes as input a trained VQC (i.e., with fixed parameters $\theta$), a test input $x$, a class label $y$, and a perturbation radius $\epsilon$. 
It encodes the input as intervals, i.e., each feature of $x$ is extended to $[x_i - \epsilon, x_i + \epsilon]$, and propagates these intervals through the encoding and a symbolic representation of the variational layers of the circuit.

During interval propagation, both real and complex-valued amplitudes are tracked using interval arithmetic. The prediction is derived by identifying the output class whose associated probability interval has a lower bound that exceeds the upper bounds of all others, ensuring the prediction is unambiguous under the entire input region. If this condition is not met, for instance if the intervals of two or more output classes overlap, an iterative refinement process is employed: the input interval with the largest uncertainty (i.e., the largest width) is split, and the verification is recursively applied to the sub-regions until a conclusive result is reached (i.e., all the subregions are safe or a single counterexample is discovered) or a specified precision is met.\footnote{We also experimented with different heuristics for selecting the node to split during refinement—for example, choosing nodes at random—but the strategy described above consistently yielded the best results in terms of precision and was therefore used to generate the reported outcomes.}
Broadly speaking, the algorithm combines an exponential and binary search. It starts from a minimal $\epsilon_{\min}$ and iteratively doubles it until the robustness condition is violated or $\epsilon_{\max}$ is reached. Then, a binary search is used to converge to the largest $\epsilon$ that still guarantees robustness, with a desired precision $\tau$.

\begin{figure}[!h]

  \begin{minipage}[t]{0.47\textwidth}
    \centering
    \vspace{-2.2cm}
    \scriptsize
    \begin{algorithm}[H]
    \caption*{Compute max $\epsilon$ perturbation}
    \begin{algorithmic}[1]
    \State \textbf{Input:} VQC with trained parameter $\theta$, input $x$, $y$ class, $\epsilon_{\min}$, $\epsilon_{\max}$ , $\tau$ tolerance threshold
    \State $\epsilon \gets \epsilon_{\min}$
    \While{$\epsilon \leq \epsilon_{\max}$ \textbf{and} $\textsc{Verify}(VQC, x, \epsilon, y)== \text{robust}$}
        \State $\epsilon \gets 2 \cdot \epsilon$
    \EndWhile
    \State $high \gets \min(\epsilon, \epsilon_{\max})$
    \State $low \gets \epsilon / 2$
    \State $max\_epsilon \gets low$
    \While{$high - low > \tau$}
        \State $mid \gets (low + high) / 2$
        \If{$\textsc{Verify}(VQC, x, \epsilon, y)== \text{robust}$}
            \State $max\_epsilon \gets mid$
            \State $low \gets mid$
        \Else
            \State $high \gets mid$
        \EndIf
    \EndWhile
    \State \Return $max\_epsilon$
    \end{algorithmic}
    \end{algorithm}
  \end{minipage}%
  \hspace{1mm}
    \begin{minipage}[c]{0.45\textwidth}
    \centering
    \vspace{2cm}
     \scriptsize
    \begin{tabular}{lcc}
    \multicolumn{2}{l}{\textbf{Robustness Verification results
    }}& \\
    \toprule
    \textbf{Model-Dataset} & \textbf{Accuracy (\%)} & \textbf{Mean of max $\epsilon$ pert.} \\
    \midrule
    QCL-Iris & 76\% & 0.07496$\pm$0.05 \\
    CCQC-Iris & 100\%  & 0.1244$\pm$0.04 \\
    PV-MNIST[0,1] & 95\%  & 0.0048$\pm$0.002 \\
    PV-MNIST[2,6] & 78\% & 0.0022$\pm$0.001 \\
    \bottomrule
    \end{tabular}

  \end{minipage}
    \caption{Left: the verification pipeline. Right: robustness verification results. The mean of the maximum $\epsilon$ is computed over 10 randomly selected inputs from the test set and the model with the highest accuracy during training.}
    \label{fig:verification_results}
\end{figure}

The right side of \autoref{fig:verification_results} reports the verification results on different quantum classifiers and datasets. 
The results reveal several interesting insights:  
\texttt{QCL-Iris} achieves an accuracy of 76\% with a mean maximum $\epsilon$ of 0.07496. This relatively high perturbation tolerance reflects the low dimensionality and structured nature of the Iris dataset. 
On the other hand, \texttt{CCQC-Iris}, which adopts a different encoding or circuit structure, achieves perfect accuracy (100\%) and significantly higher robustness (mean $\epsilon = 0.1244$). 
This indicates better alignment between circuit expressivity and the problem structure, possibly due to improved use of parameterized gates.  
For the more complex MNIST dataset, particularly in binary classification tasks over digits [0,1] and [2,6], we observe a clear drop in robustness: $0.0048$ and $0.0022$, respectively. 
While \texttt{PV-MNIST[0,1]} still maintains high accuracy (95\%), the low maximum $\epsilon$ highlights how small input perturbations can significantly alter the classification outcome in high-dimensional settings. 
This is even more evident for \texttt{PV-MNIST[2,6]}, where both accuracy and robustness are lower, suggesting that distinguishing these digits is harder for the model.

Overall, the results validate the effectiveness of our reachability-based verification approach in quantifying model robustness in quantum machine learning. The method not only provides formal guarantees but also reveals the limitations of current quantum models in handling high-dimensional, less separable data. These insights could drive the design of more robust circuits and encoding schemes in future work.


\section{Related Work}

\paragraph{NN verification}
Classical NN formal verification tools \cite{liu2021algorithms,wei2024modelverification}, such as those developed in \cite{Ai2} and \cite{singh2019abstract}, adopt abstract interpretation techniques to conservatively approximate the network’s behavior by propagating abstract domains (like intervals or zonotopes) through its layers. As demonstrated in this work, this approach allows efficient and sound verification of properties—such as robustness—by computing conservative bounds on the output set. A different category of tools, including CROWN, $\alpha$-CROWN, and $\beta$-CROWN \cite{crown,acrown,bcrown}, employs linear relaxation methods to approximate non-linear activations with tight linear bounds, offering a balance between scalability and precision. These techniques are often integrated with Branch-and-Bound frameworks \cite{bab}, such as MN-BaB \cite{MN-BaB}, which partition the verification task into smaller subproblems—either by splitting the input perturbation space \cite{reluval} or by dividing ReLU activations into linear segments \cite{bab,babSplitRelu}. This combination enhances completeness and precision, though at the cost of increased computational complexity. More recently, emerging approaches have extended the scope of verification to identify all input regions that satisfy a given output specification \cite{dathathri2019inverse,CountingProVe,eProve,kotha2023provably,zhang2024provable}.

\paragraph{VQC Robustness Evaluation}
Robustness evaluation of variational quantum circuits remains largely unexplored. However, recent works have started to address this challenge. For example, QuanTest~\cite{quantest} proposes an adversarial testing framework that generates inputs maximizing a quantum entanglement adequacy criterion while exposing erroneous behaviors, using a gradient-based joint optimization approach. Other studies~\cite{lu2020quantum, wendlinger2024comparative} have investigated the vulnerability of VQC models to adversarial attacks, demonstrating that quantum classifiers, like their classical counterparts, can be misled by imperceptible perturbations on both classical and quantum inputs across a variety of tasks.
Nonetheless, none of these works provides a framework for establishing provable robustness guarantees for a given VQC model. To the best of our knowledge, we are the first to introduce this important capability to the quantum machine learning literature.


\paragraph{Abstract Interpretation in Quantum Computing}
Yu et al.~\cite{yu2021quantum} use abstract interpretation to reduce quantum state dimensionality by partitioning an $n$-qubit state into $m$ lower-dimensional abstractions over $k < n$ qubits, enabling assertion verification in static circuits.
Feng et al.~\cite{feng2023abstract} connect quantum Hoare logic and quantum incorrectness logic to abstract interpretation, showing how each framework can derive the others.
Abstract interpretation has also been used to analyze entanglement~\cite{perdrix_quantum_2008, honda2015analysis, Assolini24domainent, Assolini25staticent}, by defining sound abstract domains and semantics for quantum programs.
Bichsel et al.~\cite{bichsel_abstraqt_2023} propose an abstract stabilizer simulator that efficiently overapproximates the effect of Clifford and non-Clifford gates, as well as measurements.  
The abstract domain uses complex intervals: quantum states are represented as sums of Pauli operators with complex coefficients, which are abstracted using interval arithmetic and abstract Pauli operators.

\section{Conclusion}\label{sec: conclusion}

In this work, we have presented a novel framework for the formal verification of variational quantum circuits, which is based on the abstract interpretation theory.
Our approach is grounded in interval abstractions, which, despite their non-relational nature and the resulting over-approximations, still enable effective verification, as demonstrated by our empirical evaluation. To enhance precision, we have introduced several strategies inspired by techniques developed for classical deep neural network verification. 

\paragraph{Future Directions.}
We have observed that classification based on the abstract execution of a VQC may suffer from overestimation errors, preventing the system from producing a concrete output. However, recent work on classical NNs \cite{GMP24,marzari2025advancing} has shown that it is possible to relax the requirement of obtaining a precise answer by instead providing an abstract robustness guarantee. This involves applying abstraction not only to the input but also to the output, allowing for different degrees of precision. The verification task then becomes a question of whether the desired property holds within a specified level of abstraction, an idea closely related to Adequacy in static analysis \cite{Mastroeni24}. We argue that this tunable approach to verification is promising and deserves further investigation.

Another promising direction for future work is to explore alternative techniques—beyond those discussed in \autoref{sec: recovery}—for recovering precision. In particular, methods such as circuit rewriting and simplification, including the use of ZX-calculus \cite{coecke2009zx,vandewetering2020zx}, offer a compelling avenue. These techniques could be applied as a pre-processing step to simplify quantum circuits, for instance, by reducing the number of phase gates in the encoding, thereby potentially enhancing the overall precision of the analysis. The key challenge lies in understanding how to effectively integrate this powerful formalism into our verification framework and adapt it to our specific context.

Alternative representations of quantum states also offer valuable opportunities. One possibility could be the path-sum formalisms~\cite{amy_formal_2019,amy_linear_2025,richard1965feynman,Amy2018QPL,qbricks}, which provide symbolic descriptions of quantum circuits that can be more concise than vector-based representations, and have been shown effective for verifying large Clifford+T circuits. 
Similarly, bit-wise simulation techniques~\cite{chen2023quantumconstant,darosa_ketquantum_2021} allow memory-efficient encoding of quantum states, especially when entanglement is limited.
We also plan to explore the use of the Ellipsoid domain~\cite{Feret04filters,cousot06astree} to capture second-order relations between the amplitudes of a quantum state.

Finally, we plan to extend our framework to support the analysis of hybrid quantum-classical architectures~\cite{havlicek_supervised_2019,mattern2021variational,matic2022quantum}, in which VQCs appear as components within larger machine learning pipelines. This includes exploring how input abstractions can propagate through both quantum and classical layers, as well as how robustness guarantees can be preserved end-to-end.

\bibliography{bibl}

\clearpage

\end{document}